\documentclass[11pt,a4paper]{article} 

\usepackage[french, english]{babel}	
\usepackage[applemac]{inputenc} 
\usepackage{enumerate}    
\usepackage{lmodern}			
\usepackage{dsfont}			
\usepackage{amsmath, amsthm, amssymb, amsthm, latexsym, bbm,bm,mathrsfs,mathtools}	
\usepackage[pdfborder={0 0 0}]{hyperref}
\usepackage[toc,page]{appendix} 

\usepackage{geometry}
\usepackage[all]{xy} 
\usepackage{lscape}
\usepackage{color}

\usepackage{epstopdf}

\newcommand{\tr}{\operatorname{tr}}
\newcommand{\Tr}{\operatorname{Tr}}
\newcommand{\diff}{\mathrm{d}}

\newcommand{\C}{\mathbb{C}}

\newtheorem{theorem}{Theorem}[section]
\newtheorem{proposition}{Proposition}[section]
\newtheorem{notation}{Notation}[section] 
\newtheorem{lemme}{Lemma}[section] 
\newtheorem{definition}{Definition}[section] 
\newtheorem{remark}{Remark}[section]

\title{The generalized master fields}
\author{Guillaume Cébron \thanks{Guillaume Cébron is supported by the ERC advanced grant ``non-commutative distributions in free probability", held by Roland Speicher.}\and Antoine Dahlqvist\thanks{Antoine Dahlqvist is supported in part by RTG 1845 and by EPSRC grant EP/I03372X/1.} \and Franck Gabriel \thanks{Franck Gabriel is supported by the ERC ``Behaviour near criticality", held by Pr. Hairer.}}

\begin{document}

\maketitle

\begin{abstract}The master field is the large $N$ limit of the Yang-Mills measure on the Euclidean plane. It can be viewed as a non-commutative process indexed by paths on the plane. We construct and study generalized master fields, called \emph{free planar Markovian holonomy fields} which are versions of the master field where the law of a simple loop can be as more general as it is possible. We prove that those free planar Markovian holonomy fields can be seen as well as the large $N$ limit of some Markovian holonomy fields on the plane with unitary structure group.
\end{abstract}

\tableofcontents{}

\section{Introduction}
In 1954, Yang and Mills considered a non-Abelian gauge theory, which has led to the elaboration of the standard model \cite{Yang1954}. The \emph{Yang-Mills measure} is then the reference measure to compute Feynman integrals. It is supposed to have a density with respect to a translation invariant measure on the space of connections on a principal bundle. Unfortunately, this definition is mathematically meaningless, since the space of connections is not locally compact. Giving a rigorous description of the Yang-Mills theory in four dimensions is still an open problem, but in the case of a two-dimensional space-time, its construction is possible. Physicists as Douglas~\cite{Douglas1995}, Gopakumar and Gross~\cite{Gopakumar1995} and mathematicians as Gross~\cite{Gross1988}, Driver~\cite{Driver1990}, Singer~\cite{Singer1995}, Sengupta~\cite{Sengupta1997} and Lévy~\cite{Levy2011} have explained how a connection on a surface can be thought as an element of the structure group for each path on this surface.

In this article, we focus our attention to the following case, under study in \cite{Levy2011}: the structure group is given by the \emph{unitary group $U(N)$} of dimension $N$ and the surface is given by the \emph{plane $\mathbb{R}^2$}. We refer to \cite{Levy2011} for a historical and mathematical expository of the Yang-Mills theory in that particular case. We only need to know that, in that situation, the planar Yang-Mill measure can be seen as the distribution of a \emph{planar Markovian holonomy field} \cite{Holonomy}, that is to say, a random process $\left(H_l\right)_{l \in {\sf L}_0}$, indexed by a set ${\sf L}_0$ of loops of the plane $\mathbb{R}^2$, with value in the group $U(N)$, and subject to some conditions (see Definition~\ref{planarmarkov}).

The so-called \emph{master field} is the limiting behaviour in large dimension $N$ of this Yang-Mills measure. It can be described with the help of the language of free probability of Voiculescu~\cite{Voiculescu1992}. In our setting, it is the limit in non-commutative distribution of $\left(H_l\right)_{l \in {\sf L}_0}$ (which implicitly depends on $N$) when $N$ tends to infinity. Thus, it is a non-commutative process $\left(H_l\right)_{l \in {\sf L}_0}$, indexed by a set ${\sf L}_0$ of loops of the plane $\mathbb{R}^2$, with value in the group of unitaries of a non-commutative probability space. It appears that the properties of being a planar Markovian holonomy field (Definition~\ref{planarmarkov}) for the Yang-Mills measure yields some structural properties for the master field. In this paper, any non-commutative process which shares those properties with the master field is called a \emph{free planar Markovian holonomy field} (see Definition~\ref{def:Levymaster}). Our first result, Theorem~\ref{Characterization Levy}, gives a complete parametrisation of the set of free planar Markovian holonomy fields, and comes together with a construction of all those free planar Markovian holonomy fields. In particular, it gives a third construction of the master field, different from~\cite{Anshelevich2012a,Levy2011}. Remark that being a free planar Markovian holonomy field is not sufficient to be the master field. Surprisingly, our second result, Theorem~\ref{th:infnorm}, says that the unique free planar Markovian holonomy field which is continuous in the operator norm is the master field: we distinguish the master field from other non-commutative processes by looking at global properties, without knowing the precise law of each loop.

The convergence of the planar Yang-Mills measure to the master field has been first rigorously proved in 2011~\cite{Anshelevich2012a}. Theorem~\ref{Existence Approximation}, the main motivation of this work, is the following result: every free planar Markovian holonomy field is the limit in non-commutative distribution of a sequence of planar Markovian holonomy fields. In other words, the non-commutative processes which share the global properties of the master field are also limits of planar Markovian holonomy fields in large dimension.

The rest of the article is organized as follows. In the next section, we give the definitions of a planar Markovian holonomy field, of a free planar Markovian holonomy field and of the master field. We also state our main results, that is to say Theorem~\ref{Existence Approximation}, Theorem~\ref{Characterization Levy} and Theorem~\ref{th:infnorm}. In Section~\ref{Sec:classification}, we study free planar Markovian holonomy fields. We give a construction of free planar Markovian holonomy fields and prove Theorem~\ref{Characterization Levy} and Theorem~\ref{th:infnorm}. Section~\ref{Sec:convergence} is devoted to prove Theorem~\ref{Existence Approximation}.

\section{Definition and main result}
In this section, we define classical and free planar Markovian holonomy fields in a parallel way, but also recall some basics of non-commutative probability which are needed to express our three main results in the last paragraph of the current section.
\subsection{Planar Markovian holonomy fields}

We are going to consider functions defined on the set of loops drawn in the plane and with value in some group.

The set of paths ${\sf P}$ in the plane is the set of rectifiable oriented curves drawn in $\mathbb{R}^{2}$ up to increasing reparametrization. Since we imposed the condition of rectifiability, any path $p\in {\sf P}$ has a length denoted by $\ell(p)$.  The set of loops based at $0$, denoted by ${\sf L}_{0}$, is the set of paths $l$ such that the two endpoints of $l$ are equal to $0$. Such a loop is called \emph{simple} when it does not intersect itself outside of its endpoints. There exists a natural structure on ${\sf L}_{0}$ given by the concatenation and the orientation-inversion of paths. 

\begin{notation}
Let $l_1$ and $l_2$ be two loops in ${\sf L}_{0}$. The concatenation of $l_1$ and $l_2$ is denoted by $l_1l_2$. Besides, $l_1^{-1}$ will denote the loop obtained by inverting the orientation of $l_1$. 
\end{notation}

\begin{definition}Let $\Gamma$ be a group and let ${\sf L}$ be a subset of ${\sf L}_0$. The set of multiplicative functions $\mathcal{M}ult({\sf L},\Gamma)$ from ${\sf L}$ to $\Gamma$ is the subset of functions in $\Gamma^{{\sf L}}$ such that for any loops $l_1, l_2, l_3 \in {\sf L}$, such that $l_1l_2 \in {\sf L}$ and $l_3^{-1} \in {\sf L}$, one has: 
\begin{align}
\label{eq1} f(l_1l_2) = f(l_2)f(l_1), \\
\label{eq2} f(l_3^{-1}) = f(l_3)^{-1}. 
\end{align}
\end{definition}

The notion of planar Markovian holonomy fields was defined in \cite{Holonomy}. In this paper, we will only consider {\em pure} planar Markovian holonomy fields. Since we will only restrain ourself to the setting of random holonomies on ${\sf L}_0$ (instead of the set of paths ${\sf P}$), let us give a definition of planar Markovian holonomy fields, slightly different, which fits well with our study.

We fix one underlying probability space $(\Omega,\mathcal{F},\mathbb{P})$. Let $G$ be a compact Lie group and let $d_G$ be a bi-invariant distance on $G$. The set $L^1(\Omega,G)$ of $G$-valued random variables on $\Omega$ is a group for the multiplication of random variables. 
\begin{definition}
\label{def:Markplanar}
A pure $G$-valued \emph{planar Markovian holonomy field} is an element $\left(H_l\right)_{l \in {\sf L}_0}$ of $\mathcal{M}ult({\sf L}_{0},L^1(\Omega,G))$ such that:\label{planarmarkov}
\begin{description}
\item[1-Invariance by area-preserving Lipschitz homeomorphisms: ] for any Lipschitz homeomorphism $\psi$ which preserves the Lebesgue measure on the plane, for any $n$-tuple of loops $l_1$, ..., $l_n$ which are sent by $\psi$ on $n$ loops in $L_0$, $\left(H_{l_1}, ..., H_{l_n}\right)$ has the same law as $\left(H_{\psi(l_1)}, ..., H_{\psi(l_n)}\right)$. 
\item[2-Independence property: ] for any simple loops $l_1$ and $l_2$ based at $0$ whose interiors ${\sf Int}(l_1)$ and ${\sf Int}(l_2)$ are disjoint, the two following families are independent:
$$\left\{H_l: l \in {\sf L}_0 \text{ whose image is in } \overline{{\sf Int}(l_1)} \right\},$$ 
$$\left\{H_l: l \in {\sf L}_0 \text{ whose image is in } \overline{{\sf Int}(l_2)}\right\}.$$
\item[3-Gauge invariance: ] for any $g \in G$, $\left(g H_{l} g^{-1}\right)_{l \in {\sf L}_0}$ has the same law as $\left( H_{l} \right)_{l \in {\sf L}_0}$.
\end{description}
\end{definition}

The difference between a pure and a general planar Markovian holonomy field appears in the independence property : a weaker version is required for planar Markovian holonomy fields \cite[Definition $4.1$]{Holonomy}. For sake of simplicity, since we will use only pure planar Markovian holonomy fields, we will call them planar Markovian holonomy fields and we will always assume that they are pure.

\begin{remark}The set $\mathcal{M}ult({\sf L}_0, G)$ can be endowed with the Borel $\sigma$-algebra $\mathcal{B}$ which is the smallest $\sigma$-algebra such that for any loops $l_{1}, ..., l_{n}$ in ${\sf L}_0$ and any continuous function $f: G^{n} \to \mathbb{R}$, the mapping $h \mapsto f\big(h(c_{1}), ..., h(c_{n})\big)$ is measurable.

In the article \cite{Holonomy}, a planar Markovian holonomy field is viewed as a probability measure $\mu$ on $(\mathcal{M}ult({\sf L}_0, G),\mathcal{B})$. Such an object can be turned into an element $\left(H_l\right)_{l \in L_0}$ of $\mathcal{M}ult({\sf L}_{0},L^1(\Omega,G))$ by setting $(\Omega,\mathcal{F},\mathbb{P})=(\mathcal{M}ult({\sf L}_0, G),\mathcal{B},\mu)$ and $H_l:\mathcal{M}ult({\sf L}_0, G)\to G$ which is given by the projection $H_l:h\mapsto h(l)$.

Conversely, let us start from a planar Markovian holonomy field $\left(H_l\right)_{l \in L_0}$ as defined in Definition~\ref{planarmarkov}. One can prove that there exists a measure $\mu$ on $\mathcal{M}ult({\sf L}_{0}, G)$ such that the canonical process of projections defined on $\left(\mathcal{M}ult({\sf L}_{0}, G),\mathcal{B}, \mu \right)$ has the same law as $(H_{l})_{l \in {\sf L}_{0}}$. Thus one can always suppose that almost surely $(H_{l})_{l \in {\sf L}_{0}}$ is in $\mathcal{M}ult({\sf L}_{0}, G)$. This allows to see a planar Markovian holonomy field either as a multiplicative function with value in $L^1(\Omega,G)$, or as a random multiplicative function with value in $G$ or as a probability measure on $\mathcal{M}ult({\sf L}_{0}, G)$.
\end{remark}

We have endowed the set of loops with an algebraic structure (concatenation and inversion), we can also endow it with a topological structure. Let us remind the reader the notion of convergence with fixed-endpoints defined by T. Lévy in \cite{Levy2010}. Let $(l_n)_{n \in \mathbb{N}}$ be a sequence of loops in ${\sf L}_{0}$. The sequence $(l_n)_{n \in \mathbb{N}}$ converges if and only if there exists a loop $l \in {\sf L}_{0}$ such that: 
\begin{align*}
d(l_n,l) = | \ell(l_n) - \ell(l)| + \inf \sup_{t \in [0,1]} | l_n(t)-l(t) | \underset{n \to \infty}{\longrightarrow} 0, 
\end{align*}
where the infimum is taken on the parametrizations of the loops $l_n$ and $l$. The left hand side defines actually a distance on ${\sf L}_0$. We will always consider continuity defined with respect to this notion of convergence. 
\begin{definition}
Let $(H_{l})_{l \in {\sf L}_0}$ be a $G$-valued planar Markovian holonomy field. It is \emph{stochastically continuous} if for any sequence of loops $(l_n)_{n\in \mathbb{N}}$ converging to $l$: 
\begin{align*}
\mathbb{E}\left[ d_G\left (H_{l_n}, H_{l}\right)\right] \underset{n \to \infty }{\longrightarrow}0, 
\end{align*}
or equivalently, $H_{l_n}$ converges in probability to $H_{l}$. 
\end{definition}

\subsection{Free planar Markovian holonomy fields}

When $G$ is a compact matrix Lie group in $M_N(\mathbb{C})$, a $G$-valued planar Markovian holonomy field can be consider as a process indexed by ${\sf L}_{0}$ with value in some non-commutative probability space in the following sense.

\begin{definition}A \emph{non-commutative probability space} $(\mathcal{A},\tau)$ consists of a unital {$*$-algebra} $\mathcal{A}$ endowed with a tracial positive linear functional $\tau:\mathcal{A}\to \mathbb{C}$ with $\tau(1)=1$.
\end{definition}
Indeed, we have $L^1(\Omega,G)\subset L^\infty(\Omega,M_N(\mathbb{C}))$. Then, the random holonomy field $(h(l))_{l\in {\sf L}_{0}}$ is a process with value in the non-commutative probability space $L^\infty(\Omega,M_N(\mathbb{C}))$ endowed with the trace $\mathbb{E}\circ tr$.

The definition of a free planar Markovian holonomy field follows the definition of a planar Markovian holonomy field viewed as a process indexed by ${\sf L}_{0}$ when the non-commutative probability space $(L^\infty(\Omega,M_N(\mathbb{C})),\mathbb{E}\circ tr)$ is replaced by an arbitrary one. We will see that this general definition is the good one in order to describe the limit of some planar Markovian holonomy fields whose dimension is growing. However, one has to find the right replacement of the different properties which caracterize a planar Markovian holonomy field.

Fortunately, the analogues of the basic notions of probability in the context of a non-commutative probability space are now deeply understood, and we refer to the books \cite{Speicher2006,Voiculescu1992} for a general description. Here are the non-commutative version of the concept of law and of the concept of independence.
 
\begin{definition}Let $(\mathcal{A},\tau)$ be a non-commutative probability space.
\begin{enumerate}
\item Let $(a_j)_{j\in J}$ be a family of $\mathcal{A}$. The \emph{non-commutative law} (or non-commutative distribution) of $(a_j)_{j\in J}$ is the map
 from the non-commutative polynomials $
\mathbb{C}\langle X_j,X_j^*:j\in J\rangle $ to $\mathbb{C}$ given by $  P \mapsto \tau( P(a_j:j\in J))$.
\item Two subalgebras $\mathcal{A}_1$ and $\mathcal{A}_2$ of $\mathcal{A}$ are\emph{ freely independent} if, for all $a_1\in \mathcal{A}_{i_1},\ldots, a_n\in \mathcal{A}_{i_n}$ such that $i_1,\ldots,i_n$ alternate between $1$ and $2$ (i.e. $i_1\neq i_2\neq \ldots \neq i_n$), we have
$\tau(a_1\cdots a_n)=0$ whenever $\tau(a_1)=\cdots=\tau(a_n)=0$. Two families $\mathcal{F}_1$ and $\mathcal{F}_2$ of elements of $\mathcal{A}$ are freely independent if the algebras they generate are freely independent.
\end{enumerate}
\end{definition}
We are now ready to state the definition of a free planar Markovian holonomy field.
\begin{definition}
\label{def:Levymaster}
Let $(\mathcal{A},\tau)$ be a non-commutative probability space and $\mathcal{A}_u$ be the group of unitaries of $\mathcal{A}$. A \emph{free planar Markovian holonomy field} is an element $(h_l)_{l\in {\sf L}_0}$ in $\mathcal{M}ult(\sf L_0,\mathcal{A}_u)$ such that
\begin{description}
\item[1-Invariance by area-preserving Lipschitz homeomorphisms : ] for any Lipschitz homeomorphism $\psi$ which preserves the Lebesgue measure on the plane, for any $n$-tuple of loops $l_1$, ..., $l_n$ which are sent by $\psi$ on $n$ loops in $L_0$, $\left(h_{l_1}, ...,h_{l_n}\right)$ has the same non-commutative law as $\left(h_{\psi(l_1)}, ...,h_{\psi(l_n)}\right)$.
\item[2-Independence property : ]for any simple loops $l_1$ and $l_2$ based at $0$ whose interiors ${\sf Int}(l_1)$ and ${\sf Int}(l_2)$ are disjoint, the two following families are freely independent:
$$\left\{h_l: l \in {\sf L}_0 \text{ whose image is in } \overline{{\sf Int}(l_1)} \right\},$$ 
$$\left\{h_l: l \in {\sf L}_0 \text{ whose image is in } \overline{{\sf Int}(l_2)}\right\}.$$
\item[3-Continuity : ] the convergence in non-commutative distribution of $h_{l_n}$ to $h_{l}$ whenever $(l_n)_{n\in \mathbb{N}}$ converges in ${\sf L}_0$ to $l$.
\end{description}
\end{definition}

The gauge invariance condition in the setting of free planar Markovian holonomy fields has no interest, since any family of $\mathcal{A}$ is invariant by conjugation in non-commutative law by any element of $\mathcal{A}_{u}$. Moreover, the continuity property has been formulated here differently because the stochastic continuity of a planar Markovian holonomy field has no rigourous equivalent in non-commutative probability. The formulation is justified by the fact that the stochastic continuity implies in particular the convergence in law of $H_{l_n}$ to $H_{l}$ whenever $(l_n)_{n\in \mathbb{N}}$ converges in ${\sf L}_0$ to $l$.

\subsection{Main results}
For all $N\geq 1$, the group of unitary matrices of size $N\times N$ is denoted by $U(N)$. Every free planar Markovian holonomy field is the limit of some $U(N)$-valued planar Markovian holonomy fields in the following sense.
\begin{theorem} \label{Existence Approximation}
Let $\left(h_l\right)_{l\in {\sf L}_0}$ be a free planar Markovian holonomy field in $\left(\mathcal{A},\tau\right)$. For all $N\geq 1$, there exists a $U(N)$-valued \emph{planar Markovian holonomy field} $({H_l}^{(N)})_{l\in {\sf L}_0}$, which is stochastically continuous, and whose non-commutative law converges  to the non-commutative law of $\left(h_l\right)_{l\in {\sf L}_0}$ in probability. In other words, for each $l\in {\sf L}_0$, one has the convergence in probability
$$\frac{1}{N}Tr({H_l}^{(N)})\underset{N}{\overset{\mathbb{P}}{\longrightarrow}}\tau(h_l).$$
\end{theorem}

This theorem is a consequence of a complete parametrisation of all free planar Markovian holonomy fields that we will explain now. A Lévy measure on the unit circle ${\mathbb{U}}$ is a measure $v$ on ${\mathbb{U}}$ such that $v(\{1\})=0$ and $\int _{\mathbb{U}}(\Re(\zeta)-1) dv(\zeta)<\infty$.

\begin{theorem} \label{Characterization Levy}
The non-commutative distributions of free planar Markovian holonomy fields are in one-to-one correspondence with the set of triplets $(\alpha,b,v)$, where $\alpha\in \mathbb{R}$, $b\geq 0$ and $v$ is a Lévy measure on the unit circle.
\end{theorem}

This parametrisation comes with a detailed construction of a free planar Markovian holonomy field for each triplet $(\alpha,b,v)$. The master field (up to a drift and a scaling of the area) is in particular the non-commutative distribution of the free planar Markovian holonomy field associated with $(\alpha,b,0)$, and we are able to provide the following property of the master field, which turns out to be quite unique. It requires the following notion of $L^\infty$-seminorm $\|\cdot\|_{L^\infty(\tau)}$ of a non-commutative probability space $(\mathcal{A},\tau)$: it is the seminorm with values in $[0,\infty]$ defined, for all $a\in \mathcal{A}$, by
\begin{align}
\label{eq:infnorm}
\|a\|_{L^\infty(\tau)}=\lim_{p\to \infty}\tau((a^*a)^p)^{1/2p}.
\end{align}
When considering non-commutative probability spaces which are von Neumann algebras with faithful, normal and tracial state $\tau$, it coincides with the operator norm. 
\begin{theorem} 
\label{th:infnorm}
A free planar Markovian holonomy field $\left(h_l\right)_{l\in {\sf L}_0}$ is the master field (up to a drift and a scaling of the area) if and only if it is continuous for the $L^\infty$-seminorm, in the sense that $\|h_{l_n}-h_{l}\|_{L^\infty(\tau)}$ tends to  $0$ whenever $(l_n)_{n\in \mathbb{N}}$ converges to $l$ in ${\sf L}_0$.
\end{theorem}

\section{Classification and construction}\label{Sec:classification}

\subsection{Free planar Markovian holonomy fields and free Lévy processes}
Let $(h_l)_{l\in {\sf L}_0}$ be a free planar Markovian holonomy field. For all continuous process $(l(t))_{t\geq 0}$ of simple loops of ${\sf L}_0$ such that for any $0 \leq s \leq  t$, ${\sf Int}(l(s))\subset {\sf Int}(l(t))$ and such that $dx({\sf Int}(l(t)))=t$, the forthcoming Proposition~\ref{loopprocess} says that the process $(h_{l(t)})_{t\geq 0}$ is a free unitary Lévy process (first considered in \cite{Biane1998}) in the following sense.

\begin{definition}Let $(\mathcal{A},\tau)$ be a non-commutative probability space and $\mathcal{A}_u$ be the group of unitaries of $\mathcal{A}$. A \emph{free unitary Lévy process} is a process $(a_t)_{t\geq 0}$ in $\mathcal{A}_u$ such that:
\begin{enumerate}
\item for $s\leq t$, the non-commutative laws of $a_s^{-1}a_t$ and of $a_{t-s}$ are the same,
\item For $s< t$ the families $\{a_s^{-1}a_t\}$ and $\{a_u:u\leq s\}$ are freely independent,
\item the non-commutative law of $a_t$ converges to $1$ when $t$ tends to $0$.
\end{enumerate}
\end{definition}

Thanks to the invariance by area-preserving Lipschitz homeomorphisms, the non-commutative law of the free unitary Lévy process $(h_{l(t)})_{t\geq 0}$ obtained that way does not depend of the particular choice of such a family $(l(t))_{t\geq 0}$. In fact, following the proof of \cite[Theorem 22]{Holonomy}, we obtain the following result.

\begin{proposition} Let $\left(l(t)\right)_{t\geq 0}$ be a continuous process of simple loops of ${\sf L}_0$ such that the interiors of the increments are disjoint and such that $dx(Int(l(t)))=t$. For all free planar Markovian holonomy field $(h_l)_{l\in {\sf L}_0}$ on $(\mathcal{A},\tau)$, $(h_{l(t)})_{t\geq 0}$ is a free unitary Lévy process, and the non-commutative law of $(h_l)_{l\in {\sf L}_0}$ is uniquely determined by the non-commutative law of the free unitary Lévy process $(h_{l(t)})_{t\geq 0}$.\label{loopprocess}
\end{proposition}
\begin{proof}The fact that $(h_{l(t)})_{t\geq 0}$ is a free unitary Lévy process is a direct consequence of Definition~\ref{def:Levymaster}. Now, let us explain why the non-commutative law $(h_l)_{l\in {\sf L}_0}$ is uniquely determined by the non-commutative law of the free unitary Lévy process $(h_{l(t)})_{t\geq 0}$. The arguments are those of \cite[Theorem 22]{Holonomy}, where the law of a process is replaced by the non-commutative distribution of a non-commutative process. For the reader convenience, we recall the two principal steps of the proof. First, the invariance by diffeomorphisms allows to deduce the non-commutative distribution of any finite collection of simple loops based at $0$. Secondly, the continuity condition of Definition~\ref{def:Levymaster} implies that this knowledge is sufficient to deduce the non-commutative distribution of any finite collection of loops, by an approximation argument.
\end{proof}

As a consequence, the classification of the non-commutative laws of free planar Markovian holonomy fields is equivalent to the classification of the non-commutative laws of free unitary Lévy processes. Here again, the characterisation of the non-commutative laws of free unitary Lévy processes has already been established in \cite{Bercovici1992}. We present here the (equivalent) description given in \cite{Cebron2015b}, when the characteristic pair is replaced by the free characteristic triplets. It allows to parametrize the free unitary Lévy processes (and so free planar Markovian holonomy fields), and will be useful in the rest of the paper.

Let $(a_t)_{t\geq 0}$ be a free unitary Lévy process in $(\mathcal{A},\tau)$.  The series
$\phi_{a_t}(z)=\sum_{n\geq 1}z^n \tau(a_t^n)$
is invertible near the origin, and we set 
$S_{a_t}(z)=\frac{1+z}{z}\phi_{a_t}^{-1}(z).$
A consequence of \cite{Bercovici1992} (see \cite{Cebron2015b} for the precise link) is the existence of $\alpha\in \mathbb{R}$, $b\geq 0$ and $v$ a Lévy measure on the unit circle ${\mathbb{U}}$ (i.e. a measure such that $v(\{1\})=0$ and $\int_{\mathbb{U}} (\Re(\zeta)-1) dv(\zeta)<\infty$) uniquely determined by the following identity
$$S_{a_t}(z)=\exp \left(-i\alpha t+(bz+b/2)t+t\int_{\mathbb{U}} i\Im(\zeta)+\frac{1-\zeta}{1+z(1-\zeta)}d v (\zeta)\right).$$
We call $(\alpha,b,v)$ the characteristic triplet of the Lévy process $(a_t)_{t\geq 0}$. Every $\alpha\in \mathbb{R}$, $b\geq 0$ and Lévy measure $v$ is the characteristic triplet of some free unitary Lévy process. Conversely, the non-commutative law of a free unitary Lévy process is determined by its characteristic triplet. For example, we have an explicit description of the first moment of the process, which turns out to be useful for the rest of the paper.

\begin{proposition}
\label{momentlibre}
Let $(a_t)_{t\geq 0}$ be a free unitary Lévy process with characteristic triplet $(\alpha,b,v)$. We have
$ \tau(a_t)= \exp\left(i\alpha t-bt/2+t\int_{\mathbb{U}} (\Re(\zeta)-1) dv(\zeta)\right).$
\end{proposition}
\begin{proof}It suffices to remark that $\tau(a_t)$ is the coefficient of $z$ in $\phi_{a_t}(z)$, and so can be read as $S_{a_t}(0)$.
\end{proof}

Combining the previous characterization of free unitary Lévy processes and Proposition~\ref{loopprocess}, we get a weaker version of Theorem~\ref{Characterization Levy} which can be formulated as follows.

\begin{theorem}Let $\left(l(t)\right)_{t\geq 0}$ be a continuous process of simple loops of ${\sf L}_0$ such that the interiors of the increments are disjoint and such that $dx(Int(l(t)))=t$.\label{charact triplet}

The non-commutative distribution of a free planar Markovian holonomy field $(h_l)_{l\in {\sf L}_0}$ on $(\mathcal{A},\tau)$ is uniquely determined by the characteristic triplet $(\alpha,b,v)$ of the free unitary Lévy process $(h_{l(t)})_{t\geq 0}$. Moreover, this characteristic triplet does not depend of the choice of $(h_{l(t)})_{t\geq 0}$.
\end{theorem}
The triplet $(\alpha,b,v)$ given by Theorem~\ref{charact triplet} is called the \emph{characteristic triplet of the free planar Markovian holonomy field}. Let us summarize Theorem~\ref{charact triplet} in Table~\ref{table:charac triplet}.

\begin{table}[h]
\centering
\begin{tabular}{ccccc}
Free P.M.H.F. & $\rightarrow$ & Free unitary Lévy processes & $\leftrightarrow$ & Characteristic triplets \\ 
$(h_{l})_{l\in {\sf L}_0}$ & $\rightarrow$ & $(h_{l(t)})_{t\geq 0}$ & $\leftrightarrow$ & $(\alpha,b,v)$ \\ 
\end{tabular} \caption{Characteristic triplet of a free planar Markovian holonomy field (free P.M.H.F.)}
\label{table:charac triplet}
\end{table}

We will see in next section that the arrow from free planar Markovian holonomy fields to free unitary Lévy processes can be reversed: for all  charateristic triplet $(\alpha,b,v)$, there exists a free planar Markovian holonomy field whose characteristic triplet is $(\alpha,b,v)$. For now, we turn our attention to a particular free planar Markovian holonomy field, which is known under the name of master field.

\begin{definition}Let $(\mathcal{A},\tau)$ be a non-commutative probability space, $\alpha \in \mathbb{R}$ and $b>0$. A \emph{master field} with unit volume $b$ and drift $\alpha$ is a free planar Markovian holonomy field whose characteristic triplet is $(\alpha,b,0)$.
\end{definition}

In case that $v=0$, the free unitary Lévy process whose characteristic triplet is $(\alpha,b,0)$ is a \emph{free unitary Brownian motion}, a particular process introduced in \cite{Biane1997a}, with some scaled speed given by $b$ and a drift given by $\alpha$. Thus, a master field is a free planar Markovian holonomy field whose free unitary Lévy process given by Proposition~\ref{loopprocess} is a free unitary Brownian motion. Here is a new characterization of free unitary Brownian motions among free unitary Lévy processes which turns out to be useful for us.

\begin{proposition}Let $(a_t)_{t\geq 0}$ be a free unitary Lévy process with characteristic triplet $(\alpha,b,v)$. The following sentences are equivalent:\label{prop:controlsqrt}
\begin{itemize}
\item the $L^\infty$-seminorm $\|a_t-1\|_{L^\infty(\tau)}$ tends to $0$ as $t$ tends to $0$;
\item the process $(a_t)_{t\geq 0}$ is a free unitary Brownian motion with speed $b$ and drift $\alpha$.
\end{itemize}
Moreover, when the above sentences are true, we have the following estimate for $t$ sufficiently small:
\begin{equation}\label{eq:controlsqrt}
\|a_t-1\|_{L^\infty(\tau)}=2\sin(\theta_t/2)\text{, where }\theta_t=|\alpha|t+\sqrt{(4-bt)bt}/2+\arccos(1-bt/2).
\end{equation}
\end{proposition}

\begin{proof}Let $(\alpha, b, v)$ be the characteristic triplet of $(a_t)_{t\geq 0}$. Recall that $(a_t)_{t\geq 0}$ is a free unitary Brownian motion if and only if $v=0$. Because $a_t$ is unitary, there exists a unique measure $\mu_t$ on $\mathbb{U}$ such that, for all $n\geq 0$, $\tau(a_t^n)=\int_{\mathbb{U}}x^n d\mu_t(x)$. The measure $\mu_t$ is known as the spectral measure of $a_t$ in $(\mathcal{A},\tau)$. Moreover, $\|a_t-1\|_{L^\infty(\tau)}$ is in this case the maximal distance between $1$ and any point in the support of $a_t$. As a consequence, the condition $\|a_t-1\|_{L^\infty(\tau)}$ tends to $0$ as $t$ tends to $0$ is equivalent to the convergence of the support of $\mu_t$ to $\{1\}$ for the Hausdorff distance of sets, as $t$ tends to $0$. Thus, it suffices to prove that the following sentences are equivalent:
\begin{itemize}
\item the support of $\mu_t$ converges to $\{1\}$ for the Hausdorff distance as $t$ tends to $0$;
\item we have $v=0$.
\end{itemize}

Let us first assume that the support of $\mu_t$ converges to $\{1\}$ for the Hausdorff distance of sets, as $t$ tends to $0$. Thanks to \cite[Corollary 3.8]{Cebron2015b}, we know that $n(1-\Re(x))d\mu_{1/n}(\omega_n x)$ converges weakly to $(1-\Re(x))dv+b/2\delta_1$ as $n$ tends to $\infty$, where $\omega_n= \frac{\tau(a_{1/n})}{|\tau(a_{1/n})|}.$ Finally, $(1-\Re(x))dv+b/2\delta_1$ is supported on $\{1\}$ and consequently, $v=0$.

Conversely, let us assume that $v=0$, or equivalently, that $S_{a_t}(z)$ is given by $\exp \left(-i\alpha t+(bz+b/2)t\right).$ \cite[Proposition 10]{Biane1997} gives the data of the support
$$\left\{e^{i\theta}:\theta\in \mathbb{R},|\theta- \alpha t| \leq \sqrt{(4-bt)bt}/2+\arccos(1-bt/2)\right\}$$
of $\mu_t$ for $t$ sufficiently small, from which we deduce that the support of $\mu_t$ converges to $\{1\}$ for the Hausdorff distance of sets as $t$ tends to $0$. Moreover, we deduce \eqref{eq:controlsqrt}
 as the maximal distance between $1$ and any point in the support of $\mu_t$ described above.
\end{proof}

\subsection{Construction}

In this subsection, we are going to show that, given a free Lévy process, one can associate a Free planar Markovian holonomy field: this will allow to reverse the left arrow in Table~\ref{table:charac triplet}.

A very convenient way to define the Yang-Mills field on two dimensional compact surfaces has been described by Thierry Lévy in \cite{Levy2010}. In a first step, it consists in defining it on finite graphs with some refinement compatibility. Then, this discrete definition is extended to all rectifiable loops by some continuity argument. By modifying the first step, F. Gabriel was able to define in \cite{Holonomy} all regular planar Markovian holonomy fields. We will see that this construction turns out to be sufficiently flexible to define also the master field, and more generally the free planar Markovian holonomy fields. 

To emphasize the parallel between the probabilistic construction (planar Markovian holonomy fields) and the non-commutative one (free planar Markovian holonomy fields), we will do them jointly in an abstract framework that encompasses the two notions of law we are dealing with. We will fix : 
\begin{itemize}
\item a group $\Gamma$ endowed with a pseudometric $d$ which is invariant by multiplications and by inversions: for any $h, g_1, g_2 \in \Gamma$, $d(hg_1,hg_2) = d(g_1h,g_2h) = d(g_1,g_2)$ and $d(g_1^{-1}, g_2^{-1}) = d(g_1,g_2)$. From now on, we will denote by $e$ the neutral element of $G$, 
\item a notion of law on $\Gamma$ given as follows.
\end{itemize}

\begin{definition}
\label{def:abstractlaw}
An abstract law is the data for any positive integer $k$ of an equivalence relation on $\Gamma^{k}$ such that for any positive integers $k$ and $k'$, for any ${\sf g^{1}} = (g^{1}_1, ... g^{1}_k)$ and ${\sf g^{2}} = (g^{2}_1, ..., g^{2}_k)$ in $\Gamma^{k}$, for any words $w_1$, ..., $w_{k'}$ in the letters $\left\{g_i,(g_i)^{-1}|i \in \{1,...,k\}\right\}$, the condition $(g^{1}_1,...,g^{1}_k) \sim (g_1^{2}, ..., g_{k}^{2})$ implies that:  
$$\left(w_1\left({\sf g}^{1},( {\sf g}^{1})^{-1}\right),..., w_{k'}\left({\sf g}^{1},( {\sf g}^{1})^{-1}\right)\right) \sim  \left(w_1\left({\sf g}^{2},( {\sf g}^{2})^{-1}\right),..., w_{k'}\left({\sf g}^{2},( {\sf g}^{2})^{-1}\right)\right).$$
\end{definition}

We will also suppose that the distance and the abstract law on $\Gamma$ are compatible: if $(g_1,g_2)$ and $(g'_1,g'_2)$ are two couples of $\Gamma$, if $(g_1,g_2) \sim (g'_1,g'_2)$ then:
\begin{align*}
d(g_1,g_2) = d(g'_1,g'_2).
\end{align*}

It was showed in \cite{Holonomy} that an important property needed, for the first step of the construction, is braidability. In order to introduce this property, we need to define the braid group. The simplest way to define this group is by using a generator-relation presentation ; one can read Section $1.2$ of \cite{Holonomy} for a geometric definition of the braid group. Let $k$ be an integer greater than 2.
 
\begin{definition}
\label{def:braid}
The braid group with $k$ strands $\mathcal{B}_{k}$ is the group with the following presentation: 
\begin{align*}
\left<\big(\beta_{i}\big)_{i=1}^{k-1} \text{\huge{\textbar} \normalsize} \forall i, j \in \{1,…, k-1\}, \begin{matrix} \mid i - j \mid = 1\Longrightarrow \beta_{i}\beta_{j}\beta_{i}= \beta_{j}\beta_{i}\beta_{j}\\\!\!\!\!\!\!\!\!\!\!\!\mid i-j \mid > 1\Longrightarrow \beta_{i} \beta_{j} = \beta_{j} \beta_{i} \end{matrix} \right>.
\end{align*}
\end{definition}

In the following, $(\beta_{i})_{i=1}^{k-1}$ will always denote the generators choosen in Definition \ref{def:braid}. An important fact about the braid group with $n$ strands $\mathcal{B}_{n}$ is that there exists a surjective morphism $\rho: \mathcal{B}_{k} \to \mathfrak{S}_k,$ which sends $\beta_i$ on the transposition $(i,i+1)$ for any $i \in \{1,...,k-1\}$. For the sake of simplicity, for any $\beta  \in \mathcal{B}_k$, we denote the permutation $\rho(\beta)$ by $\sigma_{\beta}$. We will need also to use some actions of the braid group on the free group $\mathbb{F}_{k}$ and on $\Gamma^{k}$. 

\begin{definition}
\label{actionlibre}
Let $\mathbb{F}_{k}$ be the free group of rank $k$ generated by $e_{1}, ..., e_{k}$. We define the natural action of $\mathcal{B}_{k}$ on $\mathbb{F}_{k}$ by: 
\begin{align*}
\beta_{i} e_{i} &= e_{i+1}, \\
\beta_{i} e_{i+1} &= e_{i+1} e_{i} e_{i+1}^{-1}, \\
\beta_{i} e_{j} &= e_{j}, \text{ for any j } \notin \{i,i+1\}. 
\end{align*} 
\end{definition} 

\begin{definition}
The natural actions of $\mathcal{B}_{k}$ and of $\mathfrak{S}_k$ on $\Gamma^{k}$ are defined by: 
\begin{align*}
\beta_{i} \bullet (x_{1},...,x_{i-1},x_{i},x_{i+1},...,x_{k}) &= \left(x_{1},...,x_{i-1},x_{i} x_{i+1} x_{i}^{-1}, x_{i},…,x_{k}\right), \\
\sigma \bullet (x_1,...,x_k) &= \left(x_{\sigma^{-1}(1)}, ..., x_{\sigma^{-1}(k)}\right), 
\end{align*}
for any $n$-tuple $(x_{i})_{i=1}^{k}$ in $\Gamma^{k}$, any integer $i\in\{1,…, k-1\}$ and any  permutation $\sigma \in \mathfrak{S}_{k}$.
\end{definition}

One can verify easily that the braid group relations are satisfied in these last definitions: the natural action of $\mathcal{B}_{k}$ on $\Gamma^{k}$ and on $\mathbb{F}^{k}$ are well defined. For a graphical computation of the braid action on $\Gamma^{k}$, one can read the discussion in \cite{Holonomy} after the Definition $7.4$.

\begin{definition}
Let ${\sf g} \in \Gamma^{k}$, we say that ${\sf g}$ is purely invariant by braids if for any braid $\beta \in \mathcal{B}_k$: 
\begin{align*}
\beta\bullet (g_1,...,g_k) \sim \sigma_{\beta} \bullet(g_1,...,g_k). 
\end{align*}
\end{definition}

In the following, a family $(g_t)_{t \geq 0}$ of elements of $\Gamma$ such that $g_{0}=e$ is called an abstract process. 
\begin{definition}
\label{def:purebraid}
A purely braidable stationary process is an abstract process $(g_t)_{t \geq 0}$ such that: 
\begin{itemize}
\item for any finite set $T$, for any family of positive real numbers $(x_t)_{t \in T}$, for any total orders $\leq$ and $\preceq$ on $T$:
\begin{align*}
\left( g_{ \sum_{t' \leq t} x_{t'}}  g_{ \sum_{t' < t} x_{t'}}^{-1} \right)_{t \in T} \sim \left( g_{ \sum_{t' \preceq t} x_{t'}}  g_{ \sum_{t' \prec t} x_{t'}}^{-1} \right)_{t \in T}, 
\end{align*}
\item for any positive integer $k$, for any $k$-tuple of reals $t_1< ...<t_k$, $\left(g_{t_k}g_{t_{k-1}}^{-1}, ..., g_{t_{1}}g_{0}^{-1}\right)$ is purely invariant by braids. 
\end{itemize}
\end{definition}

The set of piecewise affine loops will be important in the first step of the construction.
\begin{definition}
A loop $l \in {\sf L}_{0}$ is piecewise affine if it is the concatenation of segments. The set of such loops is denoted by ${\sf L_{aff}}$. 
\end{definition}

We will need also to take projective limits of some families of elements of $\Gamma$.

\begin{definition}
Let us denote by $\mathcal{P}_{f}({\sf L_{aff}})$ the set of finite subsets $F$ of ${\sf L_{aff}}$. Let us consider for any $F \in \mathcal{P}_{f}({\sf L_{aff}})$, a family $\left({\sf H}_{F}(l)\right)_{l \in F} \in \mathcal{M}ult(F,\Gamma)$. The family: 
\begin{align*}
\left(\left({\sf H}_{F}(l)\right)_{l \in F}\right)_{F \in \mathcal{P}_{f}({\sf L_{aff}})}
\end{align*}
is a projective family if for any finite subsets $F_1$ and  $F_2$ in $\mathcal{P}_{f}({\sf L_{aff}})$, if $F_1 \subset F_2$, then: 
\begin{align*}
\left({\sf H}_{F_1}(l)\right)_{l \in F_1} \sim \left({\sf H}_{F_2}(l)\right)_{l \in F_1}. 
\end{align*}
\end{definition}

In the following, we will always suppose that $(\Gamma, \sim)$ has the following projective property. 

\begin{definition}
\label{def:proj}
The group $\Gamma$ endowed with an abstract law satisfies the projective property if for any projective family of multiplicative functions $\left( \left({\sf H}_{F}(l)\right)_{l \in F}\right)_{F \in \mathcal{P}_{f}({\sf L_{aff}})}$, there exists a complete metric group $\left( \overline{\Gamma} , \overline{d}\right)$, endowed with an abstract law $\overline{\sim}$ such that: 
\begin{itemize} \item there exists a Lipschitz homeomorphism $\iota : \Gamma \to \overline{\Gamma}$ such that $\overline{d} \circ (\iota \times \iota) = d $, 
\item there exists a family $(\overline{\sf H}(l))_{l \in {\sf L_{aff}}}$ in $\mathcal{M}ult\left({\sf L_{aff}},\overline{\Gamma}\right)$ such that for any family $F\in \mathcal{P}_{f}({\sf L_{aff}})$, 
\begin{align*}
\left(\overline{\sf H}(l)\right)_{l \in F} \overline{\sim} \Big(\iota\big({\sf H}_{F}(l)\big)\Big)_{l \in F}, 
\end{align*}
\item translations and inversion are isometries on $\overline{\Gamma}$, 
\item the distance ${\overline{d}}$ is compatible with the abstract law $\overline{\sim}$: let $\{(g_1,g_2),(g_1',g_2')\} \subset \overline{\Gamma}^{2}$, if $(g_1,g_2)\overline{\sim} (g_1',g_2')$ then $\overline{d}(g_1,g_2) = \overline{d}(g_1',g_2')$,
\item for any integer $k$, for any sequences $(a_n)_{n\in \mathbb{N}}$ and $(b_n)_{n \in \mathbb{N}}$ of elements of $\overline{\Gamma}^{k}$, if $a_n$ and $b_n$ converges respectively to $a$ and $b$ when $n$ goes to infinity, and if for any positive integer $n$, $a_n \overline{\sim}{b_n}$ then $a\overline{\sim}b$. 

\end{itemize}
\end{definition}

Let us describe, in two steps, the abstract construction which will be applied after in order to construct planar Markovian holonomy fields and free planar Markovian holonomy fields.

\subsubsection{On finite planar graphs}
Let us consider a finite planar graph $\mathbb{G}$. This is the data of set of paths, the edges, $\mathbb{E}$, such that: 
\begin{itemize}
\item $\mathbb{E}$ is stable by inversion, 
\item two edges which are not each other's inverse meet, if at all, only at some of their endpoints, 
\item the bounded faces, which are the bounded connected components of $\mathbb{E}\setminus \bigcup\limits_{e \in \mathbb{E}} e\left([0,1]\right)$ are homeomorphic to an open disk. 
\end{itemize}
The set of vertices of $\mathbb{G}$ is by definition the set $\mathbb{V} = \bigcup\limits_{e  \in \mathbb{E}} e(0)$. From now on, we will always suppose that for any finite planar graph that we consider, $0 \in \mathbb{V}$.  A path in $\mathbb{G}$ is a concatenation of  its edges. It is called \emph{reduced} if it does not contain any sequence of the form $e e^{-1}$, with $e\in \mathbb{E}.$ Any finite family of reduced loops based at $0$ and drawn on $\mathbb{G}$, namely $\{l_1,...,l_k\}$, can be seen as a subset of $\pi_1(\mathbb{G})$, the fundamental group of $\mathbb{G}$ based at $0$. This group is a free group and a useful family of free generating subsets of $\pi_1(\mathbb{G})$ is given in Section 6.3 of \cite{Holonomy}. Let us discuss briefly how to construct these free generating subsets.

Let us consider $T$ a covering tree of $\mathbb{G}$. For any vertices $u$ and $v$ of $\mathbb{G}$, let us denote by $[u,v]_{T}$ the unique path in $T$ which goes from $u$ to $v$. Let us consider, for any bounded face $F$ of $\mathbb{G}$, a loop $c_{F}$ which surrounds the face $F$ in the anti-clockwise orientation and which starting point is denoted by $\underline{c}_{F}$. Let us denote by $\mathbb{F}^{b}$ the subset of bounded faces of $\mathbb{G}$. We consider the family of loops $\left({\sf l}^{T,(c_F)_{F}}_{F}\right)_{F \in \mathbb{F}^{b}}$ such that for any $F \in \mathbb{F}^{b}$: 
\begin{align*}
{\sf l}^{T,(c_F)_{F}}_{F} = \left[0,\underline{c_{F}}\right]_T c_{F} \left[\underline{c_{F}},0\right]_T.
\end{align*}

Let us consider the family $\mathcal{F}(\mathbb{G})$ of free generating subsets of $\pi_1\left(\mathbb{G}\right)$ constructed this way. An element of $\mathcal{F}(\mathbb{G})$ will be denoted by a subset $\left\{b_F, F\in \mathbb{F}^{b}\right\}$ where one has to understand that $b_{F}$ is the loop associated to the bounded face $F$. The most important properties of $\mathcal{F}(\mathbb{G})$, proved in Proposition 6.1 and Proposition 7.1 of \cite{Holonomy}, are given in the following theorem. 
\begin{theorem}
\label{th:base}
The following conditions hold: 
\begin{enumerate}
\item  for any $\{ b_F , F \in \mathbb{F}^{b}\} \in \mathcal{F}(\mathbb{G})$, $\{ b_F , F \in \mathbb{F}^{b}\}$ is a free generating subset of $\pi_{1}(\mathbb{G})$, 
\item for any two elements $\{ b_F , F \in \mathbb{F}^{b}\}$ and $\{ c_F , F \in \mathbb{F}^{b}\}$ in $\mathcal{F}(\mathbb{G})$, for any enumeration of the bounded faces $\left(F_1,...,F_{\#\mathbb{F}^{b}}\right)$,  there exists a braid $\beta \in \mathcal{B}_{k}$ such that: $$\left(b_{F_{\sigma_{\beta}(1)}},...b_{F_{\sigma_{\beta}(\#\mathbb{F}^{b})}}\right) = \beta \bullet \left(c_{F_1},...,c_{F_{\#\mathbb{F}^{b}}}\right),$$ 
\item for any Lipschitz homeomorphism $\Psi$ which sends $\mathbb{G}$ on $\mathbb{G}'$ and such that $\Psi(0) = 0$, for any $\left\{ b_F , F \in \mathbb{F}^{b}\right\}$ in $\mathcal{F}(\mathbb{G})$, the family $\left\{ \Psi(b_F) , F \in \mathbb{F}^{b}\right\}$ is in $\mathcal{F}(\mathbb{G}')$. 
\end{enumerate}
\end{theorem}

Let us consider a finite planar graph $\mathbb{G}$ with $k$ bounded faces. Let us consider an enumeration of the $k$ bounded faces of $\mathbb{G}$ : $(F_1,...,F_k)$. Let $(g_1,...,g_k)$ be a purely invariant by braids $k$-tuple of $\Gamma$.  

In the following, we will show that one can construct in some sense a canonical element of ${\sf Hom}(\pi_1(\mathbb{G}), \Gamma^{op})$, where $\Gamma^{op}$ is the group which underlying set is the same as $\Gamma$ and which product is given by $x._{op} y = yx$. Let us choose ${\sf l}$ an element of $\mathcal{F}(\mathbb{G})$: using the enumeration of the faces, one can consider ${\sf l}$ as a sequence $(l_1,...,l_k)$. By the first condition of Theorem \ref{th:base}, ${\sf l}$ is a free generating family of $\pi_1(\mathbb{G})$, thus, for any loop $l \in \pi_1(\mathbb{G})$, we can find a unique reduced word $w$ in $\left\{ l_1,...,l_k, l_1^{-1}, ..., l_k^{-1} \right\}$ such that $l = w(l_1,...,l_k, l_1^{-1}, ..., l_k^{-1})$ : we define ${\sf H}_{{\sf l}, {\sf g}}(l) = w^{op}\left(g_1,...,g_k, g_1^{-1}, ..., g_k^{-1}\right)$ where $w^{op}$ is the word $w$ read from left to right.

\begin{theorem}
The law of ${\sf H}_{{\sf l}, {\sf g}}$ does not depend on the choice of ${\sf l}$. For any other ${\sf l'} =( l_1',...,l_k') \in \mathcal{F}(\mathbb{G})$, for any positive integer $r$, for any $c_1$, ..., $c_{r}$ in $\pi_{1}(\mathbb{G})$:
\begin{align*}
\left( {\sf H}_{{\sf l}, {\sf g}} ( c_1),..., {\sf H}_{{\sf l}, {\sf g}}(c_{r}) \right) \sim \left({\sf H}_{{\sf l'}, {\sf g}} ( c_1),..., {\sf H}_{{\sf l'}, {\sf g}}(c_{r})\right). 
\end{align*}
\end{theorem}

\begin{proof}
Let us consider ${\sf l'} =( l_1',...,l_k') \in \mathcal{F}(\mathbb{G})$. Since ${\sf H}_{{\sf l}, {\sf g}}$ and ${\sf H}_{{\sf l'}, {\sf g}}$ are Lipschitz homeomorphisms and using the definition of an abstract law, it is enough to prove that: 
\begin{align*}
\left({\sf H}_{{\sf l}, {\sf g}}\left(l_1' \right), ..., {\sf H}_{{\sf l}, {\sf g}}\left(l_k' \right)\right) \sim \left({\sf H}_{{\sf l'}, {\sf g}}\left(l_1' \right), ..., {\sf H}_{{\sf l'}, {\sf g}}\left(l_k' \right)\right).
\end{align*}
Using the condition $2$ of Theorem \ref{th:base}, we know that there exists a braid $\beta \in \mathcal{B}_k$ such that: 
\begin{align*}
\left(l_{\sigma_{\beta}(1)}',...,l_{\sigma_{\beta}(k)}'\right) = \beta \bullet (l_1,...,l_k).
\end{align*}
Using the multiplicativity property of ${\sf H}_{{\sf l}, {\sf g}}$ and Lemma $7.1$ of \cite{Holonomy}: 
\begin{align*}
\sigma_{\beta}^{-1} \bullet \left({\sf H}_{{\sf l}, {\sf g}}\left(l_1' \right), ..., {\sf H}_{{\sf l}, {\sf g}}\left(l_k' \right)\right) = \beta^{-1} \bullet (g_1,...,g_k).
\end{align*}
Since ${\sf g}$ is purely invariant by braids, $\sigma_{\beta} \bullet (\beta^{-1} \bullet (g_1,...,g_k)) \sim (g_1,...,g_k)$, and since by definition, 
\begin{align*}
 (g_1,...,g_k) = \left({\sf H}_{{\sf l'}, {\sf g}}\left(l_1' \right), ..., {\sf H}_{{\sf l'}, {\sf g}}\left(l_k' \right)\right),
\end{align*}
we get the following fact: 
\begin{align*}
\left({\sf H}_{{\sf l}, {\sf g}}\left(l_1' \right), ..., {\sf H}_{{\sf l}, {\sf g}} \left(l_k' \right)\right) \sim  \left({\sf H}_{{\sf l'}, {\sf g}}\left(l_1' \right), ..., {\sf H}_{{\sf l'}, {\sf g}}\left(l_k' \right)\right),
\end{align*}
which allows us to conclude the proof. 
\end{proof}

From now on, let us consider ${\sf g} = \left(g_t\right)_{t \geq 0}$ a purely braidable stationary process in $\Gamma$. Let ${\sf dx}$ be the Lebesgue measure on the plane. 
Let us consider for any $i \in \{1,...,k\}$, 
\begin{align*}
h_{i} = g_{\sum_{j=1}^{i}{\sf dx}(F_j)} \left(g_{\sum_{j=1}^{i-1}{\sf dx}(F_j)}\right)^{-1},
\end{align*}
and ${\sf h} = (h_1,...,h_k)$. The $k$-tuple ${\sf h}$ is purely invariant by braids. Let us choose a family ${\sf l}$ in $\mathcal{F}(\mathbb{G})$. Since the law of ${\sf H}_{{\sf l}, {\sf h}}$ does not depend on the choice of ${\sf l}$, and since we only care about the laws of the objects, we can define: 
\begin{align*}
{\sf HF}^{{\sf g}}_{\mathbb{G}} = {\sf H}_{{\sf l}, {\sf h}}.
\end{align*}

One does not forget that we chose, since the beginning, an enumeration of the $k$ bounded faces of $\mathbb{G}$. Actually, the choice of enumeration does not matter. 

\begin{lemme}
The law of ${\sf HF}^{\sf g}_{\mathbb{G}}$ does not depend on the choice of enumeration of the $k$ bounded faces of $\mathbb{G}$. 
\end{lemme}

\begin{proof}
Let us suppose that we have two enumerations of the bounded faces of $\mathbb{G}$: this means that we have two total orders $\leq$ and $\preceq$ on $\mathbb{F}^{b}$. Using the condition $1$ of the definition of a purely braidable stationary process, we know that: 
\begin{align*}
\left( g_{ \sum_{F' \leq F} {\sf dx}(F')}  g_{ \sum_{F' < F} {\sf}{\sf dx}({F'})}^{-1} \right)_{F \in \mathbb{F}^{b}} \sim \left( g_{ \sum_{F' \preceq F} {\sf dx}(F')}  g_{ \sum_{F' \prec F} {\sf dx}(F')}^{-1} \right)_{F \in \mathbb{F}^{b}}. 
\end{align*}
This allows us to see that the law of ${\sf HF}^{\sf g}_{\mathbb{G}}$ does not not depend on the choice of enumeration. 
\end{proof}

Since the law does not depend on the choice of the enumeration, we will pick an enumeration for any given finite planar graph, but we will forget about this choice. 

The Lipschitz homeomorphisms that we have just constructed are compatible for different graphs. Let us consider two graphs $\mathbb{G}_1$ and $\mathbb{G}_2$ such that $\mathbb{G}_2$ is finer than $\mathbb{G}_1$: this implies that $\pi_1\left(\mathbb{G}_1\right) \subset \pi_1\left(\mathbb{G}_2\right)$.

\begin{theorem}
\label{th:compatibilite}
For any integer $k$, for any $k$-tuple of loops $(l_1,...,l_k)$ in $\pi_1\left(\mathbb{G}_1\right)$, 
\begin{align*}
\left({\sf HF}^{\sf g}_{\mathbb{G}_1} \left( l_1\right),..., {\sf HF}^{\sf g}_{\mathbb{G}_1} \left( l_k\right)\right) \sim \left({\sf HF}^{\sf g}_{\mathbb{G}_2} \left( l_1\right),..., {\sf HF}^{\sf g}_{\mathbb{G}_2} \left( l_k\right)\right).
\end{align*}
\end{theorem}

\begin{proof}
The proof follows exactly the one of the compatibility condition in Proposition $8.2$ in \cite{Holonomy}. Let us only explain the case when $\mathbb{G}_1$ has a unique bounded face, ${\sf F}$, and $\mathbb{G}_2$ has two bounded faces in ${\sf F}$, that we denote by ${\sf F}_{r}$ and ${\sf F}_{l}$ such that $0$ is in the boundary of ${\sf F}_{r}$ and ${\sf F}_{l}$. Let $l$ be the unique loop based at $0$ in $\pi_1(\mathbb{G}_1)$ which surrounds the face ${\sf F}$ in the anticlockwise orientation. Let us prove that: 
\begin{align*}
{\sf HF}_{\mathbb{G}_1}^{{\sf g}} (l) \sim {\sf HF}_{\mathbb{G}_2}^{{\sf g}} (l).
\end{align*}
Let us consider $l_{r}$ and $l_{l}$ two loops based at $0$ such that:
\begin{enumerate}
\item $l_{r}$ is based at $0$ and surrounds ${\sf F}_{r}$ in the anticlockwise orientation,
\item $l_{l}$ is based at $0$ and surrounds ${\sf F}_{l}$ in the anticlockwise orientation,
\item $l=l_{r}l_{l}$ in $\pi_{1}(\mathbb{G}_2)$.
\end{enumerate} 
Then, 
\begin{align*}
{\sf HF}_{\mathbb{G}_2}^{{\sf g}} (l) = {\sf HF}_{\mathbb{G}_2}^{{\sf g}} (l_l){\sf HF}_{\mathbb{G}_2}^{{\sf g}} (l_r). 
\end{align*}
In order to study the law of ${\sf HF}_{\mathbb{G}_2}^{{\sf g}} (l_l){\sf HF}_{\mathbb{G}_2}^{{\sf g}} (l_r)$ we can choose any enumeration of the faces, thus we take the enumeration: $({\sf F}_{r}, {\sf F}_{l})$. This gives: 
\begin{align*}
{\sf HF}_{\mathbb{G}_2}^{{\sf g}} (l_l){\sf HF}_{\mathbb{G}_2}^{{\sf g}} (l_r) \sim \left(g_{{\sf dx}({\sf F}_r)+ {\sf dx}({\sf F}_l)}g^{-1}_{{\sf dx}({\sf F}_r)}\right)  g_{{\sf dx} ({\sf F}_{r})} = g_{{\sf dx}({\sf F}_r)+ {\sf dx}({\sf F}_l)} = g_{{\sf dx}({\sf F})}\sim {\sf HF}_{\mathbb{G}_{1}}^{\sf g} (l),
\end{align*}
where we used the fact that ${\sf dx}({\sf F}_r) +{\sf dx}({\sf F}_l) = {\sf dx}({\sf F})$.
\end{proof}

Let us consider a finite subset  $F=\{l_1,...,l_k\}$ of ${\sf L_{aff}}$. A fairly simple but important lemma is the following. 
\begin{lemme}
\label{lemme:graphe}
There exists a finite planar graph $\mathbb{G}$, whose edges are piecewise affine, such that any loop in $\{l_1,..., l_k\}$ is a concatenation of edges of $\mathbb{G}$. 
\end{lemme}
Let  $\mathbb{G}$ be such a graph. We define the family: $$\left({\sf HF}^{\sf g}_{F}( l_1), ..., {\sf HF}^{\sf g}_{F}( l_k)\right) = \left( {\sf HF}^{\sf g}_{\mathbb{G}}(l_1),...,  {\sf HF}^{\sf g}_{\mathbb{G}}(l_k) \right).$$
Using the Theorem \ref{th:compatibilite}, the law of  $\left({\sf HF}^{\sf g}_{F}(l_1), ..., {\sf HF}_{F}^{\sf g}(l_k)\right)$ does not depend on the choice of graph $\mathbb{G}$. Thus, we have constructed, for any finite subset $F$ of ${\sf L_{aff}}$, a family of multiplicative functions $\left({\sf HF}^{\sf g}_{F}(l)\right)_{l \in {\sf L_{aff}}}$. Again, the Theorem \ref{th:compatibilite} allows us to see that:
\begin{align*}
\Big\{ \left({\sf HF}^{\sf g}_{F}(l)\right)_{l \in F} , F\in \mathcal{P}_{f}\left( {\sf L_{aff}}\right) \Big\}
\end{align*}
is a projective family. Since we supposed that the group $\Gamma$ satisfies the projective property, defined in Definition \ref{def:proj}, there exists a complete metric group $\left(\overline{\Gamma}, \overline{d}\right)$ endowed with an abstract law $\overline{\sim}$, an homeomorphism $\iota : \Gamma \to \overline{\Gamma}$ and a family 
\begin{align*}
\left({\sf HF}^{\sf g}(l)\right)_{l \in {\sf L_{aff}}} \in \mathcal{M}ult({\sf L_{aff}}, \overline{\Gamma})
\end{align*}
such that for any finite family $F$ of loops in ${\sf L_{aff}}$, 
\begin{align*}
({\sf HF}^{\sf g}(l))_{l \in F} \overline{\sim} \left(\iota \left({\sf HF}_{F}^{\sf g}(l)\right)\right)_{l \in F}. 
\end{align*}

Thus for any purely braidable stationary process ${\sf g} = (g_t)_{t \geq 0}$ in $\Gamma$, we managed to construct an element ${\sf HF}^{\sf g}$ of $\mathcal{M}ult({\sf L_{aff}}, \overline{\Gamma})$, such that for any finite planar graph $\mathbb{G}$ whose edges are piecewise affine, for any family of loops $(l_1,...,l_k) \in \pi_1(\mathbb{G})$, 
\begin{align}
\label{eq:proj}
\left( {\sf HF}^{\sf g}(l_1), ..., {\sf HF}^{\sf g}(l_k) \right) \overline{\sim} \left( \iota\left( {\sf HF}_{\mathbb{G}}^{\sf g}(l_1) \right) , ...,  \iota\left( {\sf HF}^{\sf g}_{\mathbb{G}}(l_k) \right)\right).
\end{align}
In the next section, we will prove that, under some analytical condition on ${\sf g}$, one can extend by continuity ${\sf HF}^{\sf g}$ in order to get an element of $\mathcal{M}ult({\sf L}, \overline{\Gamma})$ such that for any finite planar graph $\mathbb{G}$, for any family of loops $(l_1,...,l_k) \in \pi_1(\mathbb{G})$, the Equation (\ref{eq:proj}) is satisfied. 

\subsubsection{On rectifiable loops}

In this section we will extend ${\sf HF}^{\sf g}$ in order to define a function in $\mathcal{M}ult({\sf L}_0, \overline{\Gamma})$. From now on, we will denote the function ${\sf HF}^{\sf g}$ defined on ${\sf L}_{{\sf aff}}$ by ${\sf HF}^{\sf g}_{\sf aff}$:  the name ${\sf HF}^{\sf g}$ will be used for the extension of ${\sf HF}^{\sf g}_{\sf aff}$. 

In order to extend ${\sf HF}^{\sf g}_{\sf Aff}$, we will need to suppose that the process ${\sf g} = (g_{t})_{t \geq 0}$ satisfies the following condition: there exists a constant $K \geq 0$ such that, for any $t \geq 0$,
\begin{align}
\label{eq:inegalite}
d(e,g_{t}) \leq K \sqrt{t}.
\end{align}
From now on, we will always suppose this analytical bound on ${\sf g}$. 

In fact, the analytical bound $(\ref{eq:inegalite})$ implies that $(g_t)_{t \geq 0}$ is continuous. Indeed, for any positive real numbers $t$ and $\epsilon$, using the invariance of $d$, $d\left(g_t, g_{t+\epsilon}\right) = d\left(e,g_{t+\epsilon} g_{t}^{-1}\right)$. Yet, using the first axiom in Definition \ref{def:purebraid}, $\left(g_t, g_{t+\epsilon}g_{t}^{-1}\right) \sim \left(g_{t+\epsilon}g_{\epsilon}^{-1}, g_{\epsilon}\right)$. Thus, $g_{t+\epsilon}g_{t}^{-1} \sim g_{\epsilon}$ and by compatibility of the distance and the abstract law, $d(e,g_{t+\epsilon}g_{t}^{-1}) = d(e, g_{\epsilon}) \leq K \sqrt{\epsilon}$.

Let us consider any simple loop $l \in {\sf L_{aff}}$ bounding a domain $D$.  By the projectivity property, ${\sf HF}^{{\sf g}}_{\sf Aff}(l) \overline{\sim} \iota\left( {\sf HF}_{\{l\}}^{\sf g}(l) \right)$. Using the second and fourth conditions in Definition \ref{def:proj}: 
\begin{align*}
\overline{d}( e,{\sf HF}_{\sf Aff}^{{\sf g}}(l)  ) = \overline{d}\left(\iota(e),  \iota\left( {\sf HF}_{\{l\}}^{\sf g}(l) \right)  \right) = d\left( e, {\sf HF}_{\{l\}}^{\sf g}(l)\right).
\end{align*}
Since, by definition, ${\sf HF}_{\{l\}}^{\sf g}(l)$ has the same law as $g_{{\sf dx}(D)}$, using the compatibility of the abstract law with the distance on $\Gamma$, $d\left( e, {\sf HF}_{\{l\}}^{\sf g}(l)\right) =d\left( e,  g_{{\sf dx}(D)}\right) $. Thus, for any simple loop $l \in {\sf L_{aff}}$ bounding a domain $D$, 
\begin{align*}
\overline{d}(e, {\sf HF}_{\sf aff}^{\sf g}(l)) \leq K \sqrt{{\sf dx} (D)}. 
\end{align*}
This allows us to apply the following theorem proved by T. Lévy in \cite{Levy2010}.

\begin{theorem}[Theorem 3.3.1 in \cite{Levy2010}]
\label{main}
Let ${\sf H} \in \mathcal{M}ult( {\sf L_{aff}}, \overline{\Gamma})$ be a multiplicative function. Assume that there exists $K \geq 0$ such that for all simple loop $l \in {\sf L_{aff}}$ bounding a domain $D$ and such that $\ell(l) \leq K^{-1}$, the following inequality holds: 
\begin{align}
\label{hold}
\overline{d}(e,{\sf H}(l)) \leq K \sqrt{{\sf dx}(D)}.
\end{align}
 Then the function ${\sf H}$ admits a unique extension as an element of $\mathcal{M}ult({\sf L}_0, \overline{\Gamma})$ which is continuous for the convergence with fixed endpoints.
\end{theorem}

This theorem allows us to extend the function ${\sf HF}^{\sf g}_{\sf Aff}$, and thus to prove the following result.

\begin{theorem}
\label{th:construction}
Let ${\sf g} = (g_t)_{t \geq 0}$ be a purely braidable stationary process satisfying \eqref{eq:inegalite}. There exists a complete metric group $\left( \overline{\Gamma}, \overline{d} \right)$, endowed with an abstract law $\overline{\sim}$, which satisfies the first and the last three conditions of Definition \ref{def:proj}, and there exists a multiplicative function ${\sf HF}^{{\sf g}}$ in $\mathcal{M}ult({\sf L}_{0}, \overline{\Gamma})$ such that:
\begin{description}
\item[1-Invariance by area-preserving Lipschitz homeomorphisms: ]for any Lipschitz homeomorphism $\psi$ which preserves the Lebesgue measure on the plane, for any $n$-tuple of loops $l_1$, ..., $l_n$ which are sent by $\psi$ on $n$ loops in $\sf{L_0}$, 
\begin{align*}
\left({\sf HF}^{\sf g}(l_1), ..., {\sf HF}^{\sf g}(l_n)\right) \overline{\sim} \left({\sf HF}^{\sf g}(\psi(l_1)), ..., {\sf HF}^{\sf g}(\psi(l_n))\right). 
\end{align*}
\item[2-Finite dimensional law: ]for any finite planar graph $\mathbb{G}$, for any enumeration of the bounded faces $(F_1,...,F_k)$, for any family of loops $(l_1,...,l_k) \in \mathcal{F}(\mathbb{G})$, 
\begin{align*}
\left({\sf HF}^{\sf g}(l_i)\right)_{i=1}^{k} \overline{\sim} \left(\iota\left(g_{\sum_{j=1}^{i}{\sf dx}(F_j)} \left(g_{\sum_{j=1}^{i-1}{\sf dx}(F_j)}\right)^{-1}\right)\right)_{i=1}^{k},
\end{align*}
\item[3-Continuity: ]the function  ${\sf HF}^{{\sf g}}$ is continuous for the convergence with fixed-endpoints. 
\end{description}
\end{theorem}

\begin{proof}
We will only sketch the proof of this theorem since it follows the same ideas used in the proof of Propositions $8.2$ and $8.4$ of \cite{Holonomy}. 

We have seen that we can extend the function ${\sf HF}^{\sf g}_{{\sf Aff}}$ in order to get a multiplicative function ${\sf HF}^{\sf g}$ in $\mathcal{M}ult({\sf L}_{0}, \overline{\Gamma})$ which is continuous for the convergence with fixed-endpoints. Besides, by definition it satisfies the finite dimensional law condition when $\mathbb{G}$ is a finite planar graph whose edges are piecewise affine. 

It remains to prove that the finite dimensional law condition is valid for any finite planar graph $\mathbb{G}$ and that the invariance by area-preserving Lipschitz homeomorphisms is satisfied by  ${\sf HF}^{\sf g}$. 

Let us consider a finite planar graph $\mathbb{G}$. We recall that we have always assumed that $0$ is a vertex of $\mathbb{G}$. Using Theorem $3.2$ in \cite{Holonomy}, one can approximate $\mathbb{G}$ by a sequence $\left( \mathbb{G}_n\right)_{n \in \mathbb{N}}$ of finite planar graphs with piecewise affine edges  in such a way that $\mathbb{G}_n$ is the image of $\mathbb{G}$ by a Lipschitz homeomorphism $\psi_n$ and for any bounded face $F$ of $\mathbb{G}$, ${\sf dx}\left(\psi_n(F)\right)$ converges to ${\sf dx}\left(F\right)$. Besides, the vertices of $\mathbb{G}_n$ are equal to the vertices of $\mathbb{G}$. If $(F_1,...,F_k)$ is an enumeration of the bounded faces of $\mathbb{G}$, and $(l_1, ..., l_n) \in \mathcal{F}(\mathbb{G})$, then for any positive integer $n$, $(\psi_n(l_1), ..., \psi_n(l_k))$ is in $\mathcal{F}(\mathbb{G}_n)$. Thus,
\begin{align*}
\Big({\sf HF^{g}} \left( \psi_n(l_i)\right)\Big)_{i=1}^{k} \overline{\sim} \left(\iota\left(g_{\sum_{j=1}^{i}{\sf dx}(\psi_n(F_j))} \left(g_{\sum_{j=1}^{i-1}{\sf dx}(\psi_n(F_j))}\right)^{-1}\right)\right)_{i=1}^{k}.
\end{align*} 
Using the continuity of the field ${\sf HF^{g}}$, the left hand side converges to $\left({\sf HF^{g}} \left( l_i \right)\right)_{i=1}^{k}$ and using the analytical condition on ${\sf g}$, the different axioms and the remark about the continuity of the process ${\sf g}$ explained after the Equation \ref{eq:inegalite}, the right hand side is converging to $ \left(\iota\left(g_{\sum_{j=1}^{i}{\sf dx}(F_j)} \left(g_{\sum_{j=1}^{i-1}{\sf dx}(F_j)}\right)^{-1}\right)\right)_{i=1}^{k}$. This implies that: 
\begin{align*}
\left({\sf HF^{g}} \left( l_i \right)\right)_{i=1}^{k} \overline{\sim} \left(\iota\left(g_{\sum_{j=1}^{i}{\sf dx}(F_j)} \left(g_{\sum_{j=1}^{i-1}{\sf dx}(F_j)}\right)^{-1}\right)\right)_{i=1}^{k}.
\end{align*}

In order to prove the invariance by area-preserving diffeomorphisms, using the continuity of the field ${\sf HF^{g}}$, it is enough to consider loops with piecewise affine edges, and thus, using Lemma \ref{lemme:graphe}, we have to prove that for any finite planar graph $\mathbb{G}$ and $\mathbb{G'}$ with $k$ bounded faces, if $\psi$ preserves the Lebesgue measure, and if $\mathbb{G}' = \psi(\mathbb{G})$, then for any loops $(l_1,...,l_n)$ in $\mathbb{G}$, based at $0$, $\left({\sf HF^{g}}(\psi(l_1)),...,{\sf HF^{g}}(\psi(l_n))\right)$ has the same law as $\left({\sf HF^{g}}(l_1),...,{\sf HF^{g}}(l_n)\right)$. This boils down to prove that for any $(l_1,...,l_k) \in \mathcal{F}(\mathbb{G})$, $\left({\sf HF^{g}}(\psi(l_1)), ..., {\sf HF^{g}}(\psi(l_k))\right)$ has the same law as $\left({\sf HF^{g}}(l_1),...,{\sf HF^g}(l_k)\right)$ which is true since $\left(\psi(l_1),...,\psi(l_k)\right) \in \mathcal{F}(\mathbb{G}')$ and for any bounded face $F$ of $\mathbb{G}$, ${\sf dx}(F) = {\sf dx}(\psi(F))$.
\end{proof}

When ${\sf g} = (g_{t})_{t \geq 0}$ is a usual or a free Lévy process, the finite dimensional law condition will imply the independence property of either Definition \ref{def:Markplanar} or Definition \ref{def:Levymaster}.

\subsubsection{Planar Markovian holonomy fields and free planar Markovian holonomy fields\label{From Levy to HF}}
In this section, we explain how one can use the abstract setting constructed in the previous sections in order to construct planar Markovian holonomy fields and free planar Markovian holonomy fields. The construction that we get for  planar Markovian holonomy fields is the same than the one done in \cite{Holonomy}.

Let us consider a continuous process of simple loops of ${\sf L}_0$ such that the interiors of the increments are disjoint and such that ${\sf dx}\left({\sf Int}(l(t))\right)=t$. Let us recall that $G$ is a compact Lie group, and $(\mathcal{A}, \tau)$ holds for a non-commutative probability space. Let $(Y_t)_{t \geq 0}$ and $(y_t)_{t \geq 0}$ be respectively a $G$-valued Lévy process and a $\mathcal{A}$-valued free Lévy process. Let us suppose that $(Y_t)_{t \geq 0}$ is invariant in law by conjugation by the group $G$.  

\begin{lemme}
The processes  $(Y_t)_{t \geq 0}$ and $(y_t)_{t \geq 0}$ are purely braidable stationary processes. 
\end{lemme}

\begin{proof}
The result for $(Y_t)_{t \geq 0}$ is a direct consequence of Proposition $7.2$ of \cite{Holonomy}. Let us consider $(y_t)_{t \geq 0}$. Since it has stationary and freely independent increments, the first property is satisfied. It remains to prove the purely invariance by braids of the increments. Yet, using again the freeness property of the increments, and the fact that the group of braids with $n$ strands is generated by the elementary braids $(\beta_i)_{i=1}^{n-1}$, it remains to prove that for any real $0<t_1<t_2$, $\left(g_{t_2}g_{t_1}^{-1}, g_{t_1}\right)$ is purely invariant by braids. This is a consequence from the fact that if $a$ and $b$ are two elements in $\mathcal{A}$, $(a,b)$ is always purely invariant by braids. Indeed, for any non-commutative monomial $P$ in $a$ and $b$: 
\begin{align*}
\tau\left(P(a,aba^{-1})\right) = \tau\left(P(aaa^{-1},aba^{-1})\right)  =  \tau\left(aP(a,b)a^{-1}\right) = \tau\left(P(a,b)\right).
\end{align*} 
Since the elementary braid $\beta_1$ generates the group of braids with two strands, for any braid $\beta \in \mathcal{B}_2$, $\beta\bullet (a,b)$ has the same law as $\sigma_{\beta} \bullet(a,b)$.  
\end{proof}
Let ${\sf L}$ be a subset of ${\sf L}_{0}$. For any loop $l \in {\sf L}$, the canonical projection defined on $\mathcal{M}ult\left({\sf L}, G\right)$:
\begin{align*}
\pi_{l} : \mathcal{M}ult\left({\sf L}, G\right)& \to G\\
h& \mapsto h(l)
\end{align*}
will simply be denoted by $l$, without any reference to the set ${\sf L}$. In the following we are going to change the space on which $(Y_t)_{t \geq 0}$ is defined in order to get the projective property as explained in Definition \ref{def:proj}. This will allow us to define the group $\Gamma$ and the process $(g_t)_{t \geq 0}$ used in Theorem \ref{th:construction}. We will do the same for the non-commutative algebra $\mathcal{A}$ and the free Lévy process $(y_t)_{t \geq 0}$.

\begin{lemme}
\label{lem:version1}
There exists a measure of probability on $\mathcal{M}ult\left((l_t)_{t \geq 0}, G\right)$, denoted by $\mathbb{P}$, such that the canonical process of projections $(l_t)_{t \geq 0}$, defined on the probability space $\left(\mathcal{M}ult\left((l_t)_{t \geq 0}, G\right), \mathbb{P}\right)$ has the same law as $(Y_t)_{t \geq 0}$. 
\end{lemme}

\begin{proof}
Let us remark that $\mathcal{M}ult\left((l_t)_{t \geq 0}, G\right)$ is simply the space $\{e\} \times G^{\{l_t, t > 0\}}$. Thus, the lemma only asserts that there exists a measure or probability on $\{e\} \times G^{\{l_t, t > 0\}}$ such that the canonical process of projection has the same law as $(Y_t)_{t \geq 0}.$ This is a well known result. 
\end{proof}

In order to have a similar result for the free Lévy process $(y_t)_{t \geq 0}$, one needs to define the notion of reduced loops. Let us only define the group ${\sf R}(l_t)_{t \geq 0}$. The group ${\sf R}(l_t)_{t \geq 0}$ is the set of paths which are concatenations of elements in $\left\{l_t, l_{t}^{-1} | t \geq 0\right\}$ and which contain no sequence of the form $l_t l_{t}^{-1}$ or $l_t^{-1} l_{t}$ for $t >0$. Such paths are said to be reduced. Any concatenation of elements in $\left\{l_t, l_{t}^{-1} | t \geq 0\right\}$ can be made reduced by deleting sequences of the form $l_t l_{t}^{-1}$ or $l_t^{-1} l_{t}$ for $t >0$: any concatenation of elements in $\left\{l_t, l_{t}^{-1} | t \geq 0\right\}$ will be seen as an element of ${\sf R}(l_t)_{t \geq 0}$. Let us remark that $\mathbb{C}\left[{\sf R}(l_t)_{t \geq 0}\right]$ can be endowed with an involution, denoted $*$, by linear extension of the application which sends $l \in{\sf R}(l_t)_{t \geq 0}$ on $l^{-1}$. 

\begin{lemme}
\label{lem:version2}
There exists a tracial positive linear functional $\tilde{\tau}$ on $\mathbb{C}\left[{\sf R}(l_t)_{t \geq 0}\right]$ such that the process $(l_t)_{t \geq 0}$, seen as a process in $\left(\mathbb{C}\left[{\sf R}(l_t)_{t \geq 0}\right], \tilde{\tau}\right)$ has the same non-commutative law as $(y_t)_{t \geq 0}$.
\end{lemme}

\begin{proof}
Let us define the function $\phi$ which sends $l_t$ on $y_t$ for any positive real $t$. We extend it by multiplication from ${\sf R}(l_t)_{t \geq 0}$ to $\mathcal{A}$ and then by linearity from $\mathbb{C}\left[{\sf R}(l_t)_{t \geq 0}\right]$ to $\mathcal{A}$. The functional $\tilde{\tau} = \tau \circ \phi$ satisfies the good properties. 
\end{proof}

In the commutative setting, we will consider $\Gamma_{\sf c}$ the group of $G$-valued measurable functions on $\mathcal{M}ult\left((l_t)_{t \geq 0}, G\right)$ (which is a measurable space when it is endowed with the Borel cylindrical $\sigma$-algebra), endowed with the $L^{1}$ pseudometric and the usual notion of law which are defined using the probability $\mathbb{P}$ given by Lemma \ref{lem:version1}.

In the non-commutative setting, we will take $\Gamma_{\sf nc} = {\sf R}(l_t)_{t \geq 0}$, endowed with the $L^{2}$ pseudometric defined by: 
\begin{align*}
d_{\Gamma_{\sf nc}}( a, b) = \sqrt{\tilde{\tau}\left( (a-b)( a-b)^{*}\right) }
\end{align*}
and the usual notion of non-commutative law given by $ \tilde{\tau}$ where we recall that the tracial functional $ \tilde{\tau}$ was given by Lemma \ref{lem:version2}.

In the commutative setting, we consider ${\sf g}$ equal to the process of canonical projections $(l_t)_{t \geq 0}$ seen as measurable functions on $\mathcal{M}ult\left((l_t)_{t \geq 0}, G\right)$. In the non-commutative setting we consider ${\sf g}$ equal to the process of loops $(l_t)_{t \geq 0}$ seen as a process in ${\sf R}(l_t)_{t \geq 0}$. Using Lemmas \ref{lem:version1} and \ref{lem:version2}, these two processes have the same law respectively as $\left(Y_t\right)_{t \geq 0}$ and $\left(y_t \right)_{t \geq 0}$. Besides the two processes satisfy the analytical bound $(\ref{eq:inegalite})$.

\begin{lemme}
There exists two positive constants $K_{\sf c}$ and $K_{\sf nc}$ such that for any $t \geq 0$: 
\begin{align*}
d_{\Gamma_{\sf c}}(e,l_{t}) &\leq K_{\sf c} \sqrt{t},\\
d_{\Gamma_{\sf nc}}(e,l_{t}) &\leq K_{\sf nc} \sqrt{t}.
\end{align*}
\end{lemme}

\begin{proof}
The result in the commutative case was proved by T. Lévy in \cite{Levy2010}, Proposition $4.3.12$. Let us consider the non-commutative case. We have to consider: 
\begin{align*}
d_{\Gamma_{\sf nc}}(e,l_{t}) = \sqrt{2\left[ 1- \Re\left(\tau(y_t)\right)\right]}.
\end{align*}
In order to finish the proof, it is enought to apply Proposition \ref{momentlibre}. 
\end{proof}

In order to apply Theorem \ref{th:construction}, we still need to show that $\Gamma_{\sf c}$ and $\Gamma_{\sf nc}$ satisfy the projective property defined in Definition \ref{def:proj}. 

\begin{proposition}
The two groups $\Gamma_{\sf c}$ and $\Gamma_{\sf nc}$ endowed with the abstract law and the distance defined above satisfy the projective property. 
\end{proposition}

\begin{proof}
Let us consider the group  $\Gamma_{\sf c}$ endowed with the notion of law coming from $\mathbb{P}$ and the $L_1$ pseudometric associated to it. Let us suppose that we have a projective family of multiplicative functions $\left( ({\sf H}_F(l))_{l \in F}\right)_{F \in \mathcal{P}_f({\sf L_{aff}})}$. Applying Proposition $2.1$ of \cite{Holonomy} or Proposition $2.2.3$ of \cite{Levy2010}, there exists a measure $\mathbb{P}_{{\sf Aff}}$ on $\mathcal{M}ult({\sf L}_{\sf aff}, G)$ such that for any $F \in \mathcal{P}_{f}(T)$ the law of $ ({\sf H}_F(l))_{l \in F}$ is the same as the law of the canonical process $(l)_{l \in F}$ on $\mathcal{M}ult({\sf L}_{\sf aff}, G)$. We consider $\overline{\Gamma_{\sf c}}$, the group of $G$-valued measurable functions on $\mathcal{M}ult({\sf L}_{\sf aff}, G)$, endowed with the $L_1$ metric and the notion of law which are defined by using the probability $\mathbb{P}_{{\sf Aff}}$. 

The translations and inversion are isometries on $\overline{\Gamma_{\sf c}}$: this is a consequence of the invariance of the distance on $G$. The distance $\overline{d}$ is obviously compatible with the notion of law and satisfies the fifth point of Definition \ref{def:proj}. Besides, there exists a natural homeomorphism $\iota: \Gamma_{\sf c} \to \overline{\Gamma_{\sf c}}$, it is induced by the restriction function $\psi:  \mathcal{M}ult\left({\sf L}_{\sf aff},G\right) \to \mathcal{M}ult\left((l_t)_{t \geq 0},G\right)$. For this homeomorphism, the relation $\overline{d}(\iota \times \iota) = d$ holds. The family $(\overline{\sf H}(l))_{l \in {\sf Aff}}$ is given by the canonical process $(l)_{l \in {\sf L_{aff}}}$ on $\mathcal{M}ult({\sf L}_{\sf aff}, G)$ which is, by tautology, an element of $\mathcal{M}ult({\sf L}_{\sf Aff},\overline{\Gamma_{\sf c}})$. By definition of $\mathbb{P}_{\sf Aff}$, the second property of Definition \ref{def:proj} is satisfied.

Let us consider the group  $\Gamma_{\sf nc}$ endowed with the notion of law coming from $\tilde{\tau}$ and the $L_2$ pseudometric associated to it. Let us consider ${\sf RL_{Aff}}$ the group of reduced  piecewise affine loops based at $0$. This is the set of loops $l$ in ${\sf L_{aff}}$ such that $l$ does not contain any sequence of the form $ee^{-1}$ where $e$ is a piecewise affine path. Again, any piecewise affine loop can be made reduced by deleting sequences of the form $ee^{-1}$, and thus any loop in ${\sf L_{Aff}}$ can be seen as an element of ${\sf RL_{Aff}}$. Let us consider a projective family of multiplicative $\Gamma_{\sf nc}$-valued functions $\left( \left( {\sf H}_{F}(l) \right) \right)_{F \in \mathcal{P}_{f}({\sf L_{aff}})}$. We can define a function on ${\sf L_{aff}}$, denoted by $\tau$, such that for any $l \in {\sf L_{aff}}$, $\tau(l) = \tilde{\tau} \left({\sf H}_{\{l\}}(l)\right)$. The linear extension of $\tau$ on $\mathbb{C}\left[{\sf RL_{aff}}\right]$ is positive: this is a consequence of the projective property and the positivity of $\tilde{\tau}$. As for $\mathbb{C}[{\sf R}(l_t)_{t \geq 0}]$, one can define an involution, denoted by $*$, on $ \mathbb{C}[{\sf RL}_{\sf aff}]$, which is the antilinear extension of the application which sends $l \in {\sf RL}_{\sf aff}$ on $l^{-1}$. Using this fact, we can define a pseudometric on ${\sf RL_{aff}}$: $d(l,l') = \sqrt{\tau((l-l')(l-l')^{*})}$. Let us consider: 
\begin{align*}
H = \left\{ l \in {\sf RL_{aff}} | d(l,e) = 0 \right\}. 
\end{align*}
Using the invariance by multiplication and by inversion of $d$, and using the triangular inequality, $H$ is a group and it is normal in ${\sf RL_{aff}}$. Thus, ${\sf RL_{aff}}/H$ is a group, and the pseudometric $d$ defines on ${\sf RL_{aff}}/H$ a distance. We define $\overline{\Gamma_{\sf nc}}$ the completion of ${\sf RL_{aff}}/H$ for the distance $d$: the distance on the completion will be denoted by $\overline{d}$. The family $(\overline{{\sf H}}(l))_{l \in {\sf L_aff}}$ will be the family $(l)_{l \in {\sf L_{aff}}}$. The homeomorphism $\iota: \Gamma_{\sf nc} \to \overline{\Gamma_{\sf nc}}$ is the natural application from $\Gamma_{\sf nc}$ in $\overline{\Gamma_{\sf nc}}$ resulting from the injection $\Gamma_{\sf nc} \to {\sf RL_{aff}}$. It remains to define an abstract law on $\overline{\Gamma_{\sf nc}}$. The application $\tau$ defined on $\mathbb{C}[{\sf RL_{aff}}]$ satisfies the Cauchy-Schwarz inequality: for any $l$ and $l'$ in ${\sf RL_{aff}}$, 
\begin{align}
\label{continuity}
\tau(l) - \tau(l') = \tau(l-l') \leq d(l,l'). 
\end{align}
Thus, if $d(l,l') = 0 $ then $\tau(l)= \tau(l')$: the function $\tau$ defines a function, also denoted by $\tau$, on ${\sf RL_{aff}}/H$ and, again by (\ref{continuity}), it defines a function on $\overline{\Gamma_{\sf nc}}$. By linearity, we extend $\tau$ on $\mathbb{C}[\overline{\Gamma_{\sf nc}}]$: $\left( \mathbb{C}[\overline{\Gamma_{\sf nc}}], \tau\right)$ is a non-commutative probability space. The abstract law on $\overline{\Gamma_{\sf nc}}$ is the restriction of the notion of non-commutative law defined on $\left( \mathbb{C}[\overline{\Gamma_{\sf nc}}], \tau\right)$. It is not difficult to see that all the properties are satisfied with these definitions. 
\end{proof}

Using all the discussion we had above, we can apply Theorem \ref{th:construction}. In the commutative case, we have thus constructed a multiplicative function ${\sf HF}^{{\sf Y}}$ in the space $\mathcal{M}ult({\sf L}_0, L^{1}(\mathcal{M}ult({\sf L}_0, G), \mathbb{P}))$ which satisfies some invariance by area-preserving Lipschitz homeomorphism property, a finite dimensional law property and a continuity property. In order to prove that it is a G-valued planar Markovian holonomy field, we must show that it satisfies the independence property of Definition \ref{def:Markplanar}. We do not explain the proof since it would follow the same arguments as the one used to prove the axiom ${\sf wDP_2}$ in the proof of Proposition $8.4$ in \cite{Holonomy}.

\begin{theorem}
Let ${\sf Y}=(Y_t)_{t \geq 0}$ be a Lévy process on $G$ which is invariant in law by conjugation by $G$. There exists a stochastically continuous planar Markovian holonomy field $\big({\sf HF}^{{\sf Y}}(l)\big)_{{\sf L}_0}$, associated with $(Y_t)_{t \geq 0}$, such that for any finite planar graph $\mathbb{G}$, for any enumeration of the bounded faces $(F_1,...,F_k)$, for any family of loops $(l_1,...,l_k) \in \mathcal{F}(\mathbb{G})$, $\big({\sf HF}^{{\sf Y}}(l_i)\big)$ is a vector of independent random variables, and for any $i \in \{1,...,k\}$, ${\sf HF}^{{\sf Y}}(l_i)$ has the same law as $Y_{{\sf dx}(F_i)}$. \label{th:injectivlevyun}
\end{theorem} 

By considering the non-commutative setting, one gets the following theorem. 

 \begin{theorem}
 Let ${\sf y}=(y_t)_{t \geq 0}$ be a free unitary Lévy process. There exists a free planar Markovian holonomy field $\big({\sf HF}^{{\sf y}}(l)\big)_{l\in {\sf L}_0}$, associated with $(y_t)_{t \geq 0}$, such that for any finite planar graph $\mathbb{G}$, for any enumeration of the bounded faces $(F_1,...,F_k)$, for any family of loops $(l_1,...,l_k) \in \mathcal{F}(\mathbb{G})$, $\big({\sf HF}^{{\sf y}}(l_i)\big)$ is a vector of freely independent variables, and for any $i \in \{1,...,k\}$, ${\sf HF}^{{\sf y}}(l_i)$ has the same law as $y_{{\sf dx}(F_i)}$. \label{th:injectivlevy}
 \end{theorem}

\subsection{Proof of Theorem~\ref{Characterization Levy} and Theorem~\ref{th:infnorm}}

Combining Theorem~\ref{charact triplet} and Theorem~\ref{th:injectivlevy} gives a proof of Theorem~\ref{Characterization Levy}.

Let us prove Theorem~\ref{th:infnorm}. First, let $\big(h_l\big)_{l\in {\sf L}_0}$ be the master field, that is to say, a free planar Markovian holonomy field with characteristic triplet $(\alpha, b, 0)$ in some non-commutative probability space $(\mathcal A,\tau)$. We consider the pseudodistance $d(x,y)=\|x-y\|_{L^\infty(\tau)}$ on the group $\mathcal A_u$ of unitaries of $\mathcal A$. Let us prove that $\big(h_l\big)_{l\in {\sf L}_{0}}\in \mathcal{M}ult( {\sf L}_{0}, \mathcal A_u)$ is continuous for this metric. By Theorem~\ref{main}, it suffices to prove that there exists $K \geq 0$ such that for all simple loop $l \in {\sf L_{aff}}$ bounding a domain $D$ and such that $\ell(l) \leq K^{-1}$, the following inequality holds: 
\begin{align}
d(1,h_l) \leq K \sqrt{{\sf dx}(D)}.
\end{align}
This is a consequence of \eqref{eq:controlsqrt}, which in our case says that, denoting by $t$ the area inside the simple loop $l$, we have, for $t$ sufficiently small,
$$d(1,h_l)=\|1-h_l\|_{L^\infty(\tau)}\leq 2\sin(\theta_t/2),$$ where $\theta_t=|\alpha|t+\sqrt{(4-bt)bt}/2+\arccos(1-bt/2)$.

Now, if $\big(h_l\big)_{l\in {\sf L}_0}$ is not the master field, which means that $v\neq 0$ in its characteristic triplet $(\alpha, b, v)$, let us prove that it is not continuous for $\|\cdot\|_{L^\infty(\tau)}$. Let $\left(l(t)\right)_{t\geq 0}$ be a continuous process of simple loops of ${\sf L}_0$ such that the interiors of the increments are disjoint and such that $dx(Int(l(t)))=t$. By Proposition~\ref{loopprocess}, we know that $(h_{l(t)})_{t\geq 0}$ is a unitary Lévy process of characteristic triplet of $(\alpha, b, v)$. Thus Proposition~\ref{prop:controlsqrt} says that $\|1-h_{l(t)}\|_{L^\infty(\tau)}$ doesn't converge to $0$ when $t$ tends $0$. On the other hand, $h_{l(t)}$ converges to $h_{l(0)}=1$ as $t$ tends $0$. Consequently, $\big(h_l\big)_{l\in {\sf L}_0}$ is not continuous for $\|\cdot\|_{L^\infty(\tau)}$.

\section{Convergence of planar Markovian holonomy fields toward free planar Markovian holonomy fields}\label{Sec:convergence}
The main purpose of this section is to prove Theorem \ref{Existence Approximation} which asserts that any free planar holonomy field can be approximated by $U(N)$-valued  planar Markovian fields. This will be achieved in section  \ref{Proof Exist Approx}.  Theorem \ref{th:injectivlevyun} and \ref{th:injectivlevy} defined one-to-one mappings $\sf{HF}$ between invariant L\'evy processes  invariant by conjugation and Markovian planar holonomy fields, as well analogues in the framework of free probability.  A first step in proving \ref{Existence Approximation} is to show that the mappings $\sf{HF}$ are continuous in the sense of non-commutative distribution. 

\begin{theorem}\label{extension of convergences}Let ${\sf Y}^{(N)}=(Y_t^{(N)})_{t\ge 0}$ be  a sequence of Lévy processes on a closed subgroup $G_N$ of $U(N)$.  Assume that

\begin{itemize}
\item For  any $t_1,\ldots,t_n\ge 0$ and $\epsilon_1,\ldots, \epsilon_n\in \{1,\ast\}$,
$\lim_{N\to \infty}\frac{1}{N}\Tr\left((Y_{t_1}^{(N)})^{\epsilon_1}\cdots (Y_{t_n}^{(N)})^{\epsilon_n}\right)$ exists in probability. 
\item There exists a constant $C>0$ such that for all $N\ge 1,$ $$1-\frac{1}{N}\mathbb{E}\left[\Re\left(\Tr(Y_t)\right)\right]\le C t. $$  
\end{itemize}
Then  there  exists a  trace $\Phi$ on $(\C[{\sf L}_0],*)$ such that  for any $l\in \sf L_0,$ 
$$\mathbb{E}\left[\left|\frac{1}{N}\Tr({\sf HF^Y}(l))-\Phi(l)\right|\right]\longrightarrow0.$$ 
If furthermore  $(\mathcal{A},\tau)$ is a non-commutative probability space containing a free unitary  L\'evy process ${\sf y}=(y_t)_{t\ge 0}$  and the free planar Markovian holonomy field $\sf{HF}^y,$ such that $$\lim_{N\to \infty}\frac{1}{N}\Tr\left((Y_{t_1}^{(N)})^{\epsilon_1}\cdots (Y_{t_n}^{(N)})^{\epsilon_n}\right)=\tau\left((y_{t_1})^{\epsilon_1}\cdots (y_{t_n})^{\epsilon_n}\right),$$ for any $t_1,\ldots,t_n\ge 0$ and $\epsilon_1,\ldots, \epsilon_n\in \{1,\ast\}$, then, for any $l\in {\sf L}_0,$ 
$$\Phi(l)= \tau({\sf HF^y}(l)).$$ 
\end{theorem}

In section \ref{Sec:Existence Approximation}, a second step   in proving theorem \ref{Existence Approximation} will be  the use of an approximation by finite unitary matrices of  any free unitary L\'evy process fulfilling the above conditions.

 \begin{remark}   Approximation results  can be  found using different classical series  for example with  orthogonal, symplectic compact groups, but also with  groups of permutations  viewed as matrices.  Then, the Theorem can also be applied. This latter example was investigated in   \cite{Gabriel2015c}, where  the above Theorem \ref{extension of convergences} was used to get  therein Theorem 3.2.
 \end{remark}
  \begin{remark}  The random variables considered being bounded, their convergence in $L^1$ and in probability are equivalent. A yet unsolved problem is to prove that the theorem \ref{extension of convergences} holds true when the convergences in $L^1$ are replaced by almost sure convergences. 
 \end{remark}
\begin{remark}The second assertion of Theorem \ref{extension of convergences} yields that the free planar Markovian holonomy fields are the only possible limiting objects in the framework of free probability.  Let us prove it   assuming the first one to hold true.
\end{remark}

We will prove Theorem \ref{extension of convergences} in two steps, corresponding to the  next two sections. 

\subsection{Convergence for affine loops}

We shall prove that the convergence of normalized traces for the L\'evy process $(Y^{N}_t)_{t\ge0}$ yields the one for the  field $({\sf HF}^{Y}(l))_{l\in \sf L_{Aff}}$. 
\begin{lemme} \label{conv affine} Let $(Y_t)_{t\ge 0}$ be a process satisfying the first conditions of Theorem  \ref{extension of convergences}, then,  there exist a  trace $\Phi$ on  $ \C[\sf L_{Aff}]$ with
$$\mathbb{E}\left[ \left|\frac{1}{N}\Tr({\sf HF}^Y(l))-\Phi(l)\right|\right]\to0.$$ If $(y_t)_{t\ge 0}$ is free unitary L\'evy process in $(\mathcal{A},\tau)$ fulfilling the second conditions then for any $l\in \sf{L}_{\sf Aff}$,
$$\Phi(l)=\tau\left({\sf HF}^y(l)\right).$$
\end{lemme}

\begin{proof}For any $l\in \sf L_{Aff}$,  let us consider the embedded graph $\mathbb{G}$ containing $l,$ given   by  Lemma \ref{lemme:graphe} for $\{l\}$.  According to point 2. of Theorem \ref{th:construction},  for any enumeration $F_1,\ldots, F_k$ of the bounded faces of $\mathbb{G},$ if ${\sf l}=(l_1,\ldots,l_k)\in\mathcal{F}(\mathbb{G}),$ then, 
$$\left({\sf HF}^Y(l_i)\right)_{i=1}^k\overset{(\text{law})}{=} \left(Y_{\sum_{j=1}^{i} dx(F_j)} \left(Y_{\sum_{j=1}^{i-1}dx(F_j)} \right)^{-1}\right)_{i=1}^k.$$  
By assumption, the right tuple converges in non-commutative distribution in probability, furthermore  towards $\left(y_{\sum_{j=1}^{i} dx(F_j)} \left(y_{\sum_{j=1}^{i-1}dx(F_j)} \right)^{-1}\right)_{i=1}^k$, if  the second condition is satisfied.  Let us decompose $l$ in  the basis $(l_1,\ldots,l_k)$ of $\pi_1(\mathbb{G})$ as a word $w(l_i,i=1..k),$  and write $w^op$ for the same word written in reverse.  Then,  $\frac{1}{N}\Tr\left({\sf HF^Y}(l)\right)=\frac{1}{N}\Tr\left(w^{op}({\sf HF^Y(l_i)},i=1..k)\right)$ converges in probability (towards $\tau\left(w^{op}({\sf HF^y}(l_i),i=1..k)\right)=\tau\left({\sf HF^y}(l)\right)$, when  the second assumption is fulfilled).
\end{proof}

\textbf{Two uniformity estimates:}  Extending this result to any  loop in $\sf L_0$ leads to the following problem of  uniform convergence.  Consider a loop $l\in \sf L_0$ and an approximation of  it  by affine loops $l_n\in \sf L_{Aff}$.  Together with  the continuity property of the  holonomy fields (3. of Theorem \ref{th:construction}), the  Lemma \ref{conv affine}  yields  the following  double limit array  of  $\sf L^1$-convergence:     
\begin{align*}
 \xymatrix{
 & \frac{1}{N}\textbf{${\sf Tr}\left({\sf HF^Y}(l_n)\right)$}\ar@/^{+0pc}/[rr]^{n \to \infty}\ar@/^{+0pc}/[dd]_{N \to \infty} &&  \frac{1}{N}\textbf{${\sf Tr}({\sf HF^Y}(l))$}\ar@{.>}[dd]^{N \to \infty} \\
\text{ } \\
 & \Phi(l_n) \ar@{.>}[rr]^{n \to \infty} && \Phi(l)  
}
\end{align*}
where plain arrows are known convergences, whereas   convergences along   dotted arrows remain to be proved. Therefor, it would be enough to prove a uniformity estimate for  one of the plain  arrows. That is, showing whether 
\begin{align}
\label{eq:tc1}
\sup_{n} \mathbb{E}\left[\left|\frac{1}{N}\Tr({\sf HF^Y}(l_n))- \Phi(l_n)\right|\right] \underset{N \to \infty}{\longrightarrow} 0
\end{align}
or
\begin{align}
\label{eq:tc2}
\sup_{N} \mathbb{E}\left[\left|\frac{1}{N}\Tr({\sf HF^Y}(l_n))- \frac{1}{N}\Tr({\sf HF^Y}(l))\right|\right] \underset{n \to \infty}{\longrightarrow} 0.
\end{align}
In \cite{Levy2011}, T. L\'evy proved a uniformity estimate for the convergence towards the master field, yielding (\ref{eq:tc1}), when ${\sf Y}$ is a Brownian motion on\footnote{or on the compact orthogonal and symplectic groups.} $U(N)$. Namely, he showed that when ${\sf Y}$ is properly scaled Brownian motion on $U(N),$ for any $l\in \sf L_{Aff},$ 
\begin{align}
\label{eq:Levyuni}
\mathbb{E}\left[\left|\frac{1}{N}\Tr({\sf HF^Y}(l_n))- \Phi(l_n)\right|\right]  \leq \frac{1}{N}\left[ \ell(l_n) e^{\frac{1}{2} \ell(l_n)^{2}} + \ell(l_n)^{2} e^{\ell(l_n)^{2}}\right], 
\end{align}
where $\ell(l)$ denotes the length of $l$. When $(l_n)_{n\ge 1}$ converges to $l$, its length converges towards the one of $l$ and in particular $\sup_n \ell(l_n)<\infty.$  Hence, the converge.

\vspace{0,5 cm}

For other  L\' evy processes a bound alike (\ref{eq:Levyuni}) seems much more difficult to obtain. To bypass this difficulty, we will  in the next section prove  (\ref{eq:tc2}) instead. 
\begin{remark} One can tackle the problem  of fluctuations around the limit or  of  moderate deviations  with the same strategy. In \cite{Dahlqvist2014a}, it has been shown that the approach of \cite{Levy2011},  can be generalized for the Brownian motion, to get convergence results for  Laplace transforms of the considered random variables, yielding  local central limit theorems. For L\'evy processes, this remains an open question. \end{remark}

\begin{proof}[Proof of Theorem \ref{extension of convergences} knowing (\ref{eq:tc2}):]    According to Lemma \ref{conv affine}, for any integer $n\ge 0,$ $\frac{1}{N}\Tr({\sf HF^Y}(l_n))$ converges to $\Phi(l_n)$ in $L^1$.  Using (\ref{eq:tc2}) yields that $(\Phi(l_n))_{n\ge 0}$  is a Cauchy sequence. Let $\Phi(l)$ be its limit. Then, (\ref{eq:tc2}) implies that $\frac{1}{N}\Tr({\sf HF^Y}(l))$ converges  towards $\Phi(l)$ in $L^1.$  Indeed, for any $n\ge 0,$

$$\begin{aligned}\sup_N\mathbb{E}\left[\left|\frac{1}{N}\Tr({\sf HF^Y}(l))-\Phi(l) \right|\right]\le \sup_N\frac{1}{N} \mathbb{E}\left[\left|\Tr({\sf HF^Y}(l))-\Tr({\sf HF^Y}(l_n)) \right|\right]\ + |\Phi(l_n)-\Phi(l)| \end{aligned}$$
Considering $n\to\infty,$ gives $\sup_N\mathbb{E}\left[\left|\frac{1}{N}\Tr({\sf HF^Y}(l))-\Phi(l) \right|\right]=0.$ When ${\sf Y}$ converges in non-commutative distribution to a free unitary L\'evy process ${\sf y},$ such that ${\sf y}$ and ${\sf HF^y}$ take their  values in $(\mathcal{A},\tau),$ then  for any $l\in {\sf L_{Aff}},$ $\Phi(l)=\tau(\sf{HF^y}(l))$.  The continuity property of Theorem \ref{th:construction} yields that this equality holds true for any $l\in {\sf L}_0.$  
\end{proof}

In order to prove (\ref{eq:tc2}), we shall use Theorem \ref{main}, together with a uniformity estimate.

\subsection{Application of an extension theorem}
Let us fix here the approximation of a loop that we shall use.  For any loop $l\in \sf L_0$,  and any integer $n\in \mathbb{N},$  let  $D_n(l)\in \sf L_{Aff}$ be the piecewise affine loop connecting consecutively $\tilde{l}(0), \tilde{l}(2^{-n}\ell(l)),\ldots, \tilde{l}((1-2^{-n})\ell(l))$ and $\tilde{l}(0)$ with segments, where $\tilde{l}$ is the parametrization of $l$ by its length. The sequence $(D_n(l))_{n\ge 0}$ converges to $l$ \cite[Proposition $1.2.12$]{Levy2010}.   
\vspace{0,5 cm}

\noindent The following estimate is the main argument in what follows.  

\begin{proposition}[Theorem 3.3.1 in \cite{Levy2010}]\label{main bound}Let $(\Gamma,d)$ be a  complete metric group, such that $d$ is invariant by translations and inversion,  $H\in \mathcal{M}ult({\sf L}_0, \Gamma),$  a multiplicative function and  a constant $K>0$ such that  
$$d(1, H_s)\le K \sqrt{ t},$$
for any simple loop $s\in {\sf L}_0$,  bordering an area $t.$   Then, for any $ n\ge 1,$ 
$$d(H_{l},H_{D_n(l)})\le K \ell(l)^{3/4}\left(\ell(l)-\ell(D_n(l))\right)^{1/4}.$$
\end{proposition}

\begin{proof} In the proof of Proposition 3.3.7 of  \cite{Levy2010}, that for any \footnote{Here, as we consider loops in the plane, there is no  bound on the length of geodesics for the approximations of $l$,  so that  we can consider any dyadic approximation $(D_n(l))_{n\ge 0}$ for all $n\in\mathbb{N}$.} $m\ge n\ge 0,$ 
$$d(H_{D_m(l)},H_{D_n(l)})\le K \ell(l)^{3/4}\left(\ell(l)-\ell(D_n(l))\right)^{1/4}.$$
According to  Theorem  \ref{main}, the function $H$ is continuous and  as $(D_m(l))_{m\ge 0}$ converges to $l$, the statement holds true. An important remark must be added. A careful inspection shows that the  constant $K$ used in this proposition is the same as the constant of Theorem 3.3.1 and the same as ours.
\end{proof}
We can now conclude the proof of   Theorem \ref{extension of convergences} with the following argument.
\begin{proof}[Proof of the $L^1$-uniform estimate (\ref{eq:tc2}):]  Let us  consider  a sequence of L\'evy processes ${\sf Y^{(N)}}$ satisfying the assumptions  of Theorem \ref{extension of convergences}, together with the associated fields ${\sf HF^Y}$, realized on a probability space $(\Omega,\mathcal{B},\mathbb{P}).$  We  define a metric $d_N$ on   the group $\Gamma_N=L^1(\Omega,G_N)$ of $G_N$-valued random variables setting for any $X,Y\in L^1(\Omega,G_N),$ 
$$d_N(X,Y)=\frac{1}{\sqrt{2N}}\mathbb{E}\left[\Tr((X-Y)(X-Y)^*)\right]^{1/2}. $$
We shall apply the proposition \ref{main bound}  to $(\Gamma_N,d_N)$.  Since $G_N$ is closed  $\Gamma_N$ is complete, besides $d_N$ is invariant by translations and inversion.  For any simple loop $s\in {\sf  L}_0$, bounding an area $t$,  ${\sf HF^Y}(s)$  has the same  law as $Y_t^{(N)}$, therefore
$$d_N(1,{\sf HF^Y}(s))^2=1-\frac{1}{2N}\mathbb{E}\left[\Tr\left({\sf HF^Y}(s)+{\sf HF^Y}(s)^*\right)\right]=1-\Re \left(\mathbb{E}\left[\frac{1}{N}\Tr(Y_t)\right]\right).$$
By assumption,  we have a bound on the right-hand-side and we get $d_N(1,{\sf HF^Y}(s))\le \sqrt{C t}.  $ We can apply Proposition \ref{main bound} to ${\sf HF^Y}$ considered as a multiplicative function of $\mathcal{M}({\sf L}_0, \Gamma_N)$. Using Cauchy-Schwarz inequality, we get for any $n$,
$$\begin{aligned}\frac{1}{N}\mathbb{E}\left[\left|\Tr\left({\sf HF^Y}(l)\right)-\Tr\left({\sf HF^Y}(D_n(l))\right)\right|\right]&\le  \sqrt{2}d_N\left({\sf HF^Y}(l),{\sf HF^Y}(D_n(l))\right)\\
&\le \sqrt{2C} \ell(l)^{3/4}(\ell(l)-\ell\left(D_n(l)\right))^{1/4}.\end{aligned}$$
This proves the limit (\ref{eq:tc2}).
\end{proof}

\subsection{Proof of  Theorem \ref{Existence Approximation}\label{Proof Exist Approx}}\label{Sec:Existence Approximation}
The following result is a crucial step in the proof of Theorem \ref{Existence Approximation}.

\begin{theorem}[Theorem 3 of \cite{Cebron2015b}]Let $(y_t)_{t\in \mathbb{R}^+}$ be a free unitary Lévy process. There exists a sequence (indexed by $N$) of Lévy processes $({Y_t}^{(N)})_{t\in \mathbb{R}^+}$ with values in $U(N)$, starting at $I_N$, and unitarily invariant, such that $({Y_t}^{(N)})_{t\in \mathbb{R}^+}$ converges in non-commutative distribution to $(y_t)_{t\in \mathbb{R}^+}$ in the following senses: \label{Existence Approximation Levy}
for all $t_1,\ldots, t_n\in \mathbb{R}^+$ and $\epsilon_1,\ldots, \epsilon_n\in \{1,\ast\}$,
$$\lim_{N\to \infty}\mathbb{E}\left[\frac{1}{N}\Tr\left((Y_{t_1}^{(N)})^{\epsilon_1}\cdots (Y_{t_n}^{(N)})^{\epsilon_n}\right)\right]=\tau\left((y_{t_1})^{\epsilon_1}\cdots (y_{t_n})^{\epsilon_n}\right);$$
and in addition, almost surely: for all $t_1,\ldots, t_n\in \mathbb{R}^+$ and $\epsilon_1,\ldots, \epsilon_n\in \{1,\ast\}$,
$$\lim_{N\to \infty}\frac{1}{N}\Tr\left((Y_{t_1}^{(N)})^{\epsilon_1}\cdots (Y_{t_n}^{(N)})^{\epsilon_n}\right)=\tau\left((y_{t_1})^{\epsilon_1}\cdots (y_{t_n})^{\epsilon_n}\right).$$
Moreover, we can also require the following equality for the first moment: for all $N\geq 1$,
$$\mathbb{E}\left[\frac{1}{N}\Tr\left({Y_t}^{(N)}\right)\right]=\tau(y_t).$$
\end{theorem}

To be precise, the last equality is not explicitely mentioned in \cite[Theorem 3]{Cebron2015b}, but can be traced out of the proof. For the reader convenience, we describe now a possible process $({Y_t}^{(N)})_{t\in \mathbb{R}^+}$ occuring in the theorem above, making the equality $\mathbb{E}\left[\frac{1}{N}\Tr\left({Y_t}^{(N)}\right)\right]=\tau(y_t)$ explicit (see also \cite{Gabriel2015b} for a second proof of Theorem \ref{Existence Approximation Levy}).

We define the  transition  semigroup $(P_t)_{t\in \mathbb{R}^+}$ of a Lévy process $(Y_t)_{{t \geq 0}}$ on $U(N)$ as follows: for all $t\in \mathbb{R}^+$, all bounded Borel function $f$ on $U(N)$ and all $U\in U(N)$, we set $P_tf(U)=\mathbb{E}[f(UY_t)]$. The \textit{generator} of $(Y_t)_{{t \geq 0}}$, is defined to be the linear operator $L$ on $C(U(N))$ such as $Lf=\lim_{t\rightarrow 0}(P_tf-f)/t$ whenever this limit exists.
In order to describe the generator of a semigroup, we shall successively  introduce in the three next paragraphs the Lie algebra $\frak{u}(N)$ of $U(N)$, a scalar product on $\frak{u}(N)$ and the notion of Lévy measure on $U(N)$.

 The unitary group $U(N)$ is a compact real Lie group of dimension $N^2$, whose Lie algebra $\frak{u}(N)$ is the real vector space of skew-Hermitian matrices: $\frak{u}(N)=\{M\in M_N(\mathbb{C}): M^*+M=0\}.$ Any $X\in \frak{u}(N)$ induces a \textit{left invariant vector field }$X^l$ on $U(N)$ defined for all $g\in U(N)$ by $X^l(g)=DL_g(Y)$ where $DL_g$ is the differential map of $h\mapsto gh$.
We consider the following \textit{inner product} on $\frak{u}(N)$:
$$(X,Y) \mapsto \left\langle X,Y\right\rangle_{\frak{u}(N)}= N\Tr (X^*Y)=-N\Tr (XY).$$
It is a real scalar product on $\frak{u}(N)$ which is invariant under the adjoint action of $U(N)$. Let us fix an orthonormal basis $\left\{X_1,\ldots,X_{N^2}\right\}$ of $\frak{u}(N)$.\label{bloubi}

It is convenient now to introduce an arbitrary auxiliary set of local coordinates around $I_N$. Let $\Re,\Im:U(N)\rightarrow M_N(\mathbb{C})$ be such that for all $U\in U(N)$, we have $\Re (U)=(U+U^*)/2$ and $\Im (U)=(U-U^*)/2i$. Note that $i \Im$ takes its values in $\frak{u}(N)$. A \textit{Lévy measure} $\Pi$ on $U(N)$ is a measure on $U(N)$ such that $\Pi(\{I_N\})=0$, for all neighborhood $V$ of $I_N$, we have $\Pi(V^c)<+\infty$ and $\int_{U(N)} \|i \Im(x)\|_{\frak{u}(N)}^2\ \Pi(\diff x)<\infty$.

The following theorem gives us a characterization of the generator of Lévy processes.
\begin{theorem}[\cite{Applebaum1993,Liao2004}]Let $(Y_t)_{t\in \mathbb{R}^+}$ be a Lévy process on $U(N)$ starting at $I_N$.
There exist an element $X_0\in \frak{u}(N)$, a symmetric positive semidefinite matrix $(x_{i,j})_{1\leq i,j \leq N^2}$ and a Lévy measure $\Pi$ on $U(N)$ such that the generator $L$ of $\mu$ is the left-invariant differential operator given, for all $f\in C^2(U(N))$ and all $h\in U(N)$, by 
\begin{equation}
Lf(h)=X_0^lf(h)+\frac{1}{2}\displaystyle\sum_{i,j=1}^{N^2}x_{i,j}X_i^l X_j^l f(h) +\int_{U(N)}f(hg)-f(h)- \left( i \Im(g)\right)^l f(h)\ \Pi(\diff g).\label{gen}
\end{equation}
Conversely, given such a triplet $(X_0, (x_{i,j})_{1\leq i,j \leq N^2}, \Pi)$, there exists a Lévy process on $U(N)$ starting at $I_N$ whose generator is given by \eqref{gen}.
\end{theorem}
The triplet $(X_0, (x_{i,j})_{1\leq i,j \leq N^2}, \Pi)$ is called the \textit{characteristic triplet} of $(Y_t)_{t\in \mathbb{R}^+}$.

Let $(y_t)_{t\in \mathbb{R}^+}$ be a free unitary Lévy process with characteristic triplet $(\alpha,b,v)$. We consider the Lévy process $({Y_t}^{(N)})_{t\in \mathbb{R}^+}$ on $U(N)$ starting at $I_N$ with characteristic triplet $\left(i\alpha \cdot I_N, b\cdot I_{N^2},  v_N\right)$, where $ \upsilon_N$ is the Lévy measure on $U(N)$ defined, for all bounded and measurable function $f$ on $U(N)$, by
$$\int_{U(N)}f\diff v_N=N\int_\mathbb{U}\int_{U(N)}f\left(g\left(\begin{array}{cccc}\zeta & 0 & \cdots & 0 \\0 & 1 & \ddots & \vdots \\\vdots & \ddots & \ddots & 0 \\0 & \cdots & 0 & 1\end{array}\right)g^*\right)\diff g \diff v(\zeta).$$
The generator associated to $\left(i\alpha \cdot I_N, b\cdot I_{N^2},  v_N\right)$ is unitarily invariant, which implies that the process $({Y_t}^{(N)})_{t\in \mathbb{R}^+}$ is unitarily invariant. Now, \cite[Theorem 7.8, Remark 7.10]{Cebron2015b} says that we have the following convergence: for all $n\geq 1$,
$$\lim_{N\to \infty}\mathbb{E}\left[\frac{1}{N}\Tr\left(({Y_t}^{(N)})^n\right)\right]=\tau(y_t^n);$$
and in addition the following convergences which hold almost surely: for all $n\geq 1$,
$$\lim_{N\to \infty}\frac{1}{N}\Tr\left(({Y_t}^{(N)})^n\right)=\tau(y_t^n).$$
As explained in \cite[Section 7.4]{Cebron2015b}, the fact that the increments of $({Y_t}^{(N)})_{t\in \mathbb{R}^+}$ are independent and unitarily invariant together with one version of the theorem of Voiculescu \cite{Voiculescu1991} imply in particular that the process converges in non-commutative distribution to a process with free increments, in both the following senses: for all $t_1,\ldots, t_n\in \mathbb{R}^+$ and $\epsilon_1,\ldots, \epsilon_n\in \{1,\ast\}$,
$$\lim_{N\to \infty}\mathbb{E}\left[\frac{1}{N}\Tr\left((Y_{t_1}^{(N)})^{\epsilon_1}\cdots (Y_{t_n}^{(N)})^{\epsilon_n}\right)\right]=\tau\left((y_{t_1})^{\epsilon_1}\cdots (y_{t_n})^{\epsilon_n}\right);$$
and in addition, almost surely: for all $t_1,\ldots, t_n\in \mathbb{R}^+$ and $\epsilon_1,\ldots, \epsilon_n\in \{1,\ast\}$,
$$\lim_{N\to \infty}\frac{1}{N}\Tr\left((Y_{t_1}^{(N)})^{\epsilon_1}\cdots (Y_{t_n}^{(N)})^{\epsilon_n}\right)=\tau\left((y_{t_1})^{\epsilon_1}\cdots (y_{t_n})^{\epsilon_n}\right).$$
Remains the last equality in Theorem \ref{Existence Approximation Levy} above, which is a consequence of \cite[Proposition 5.8]{Cebron2015b}. Alternatively, one can argue as follows. Denoting by $L_N$ the generator of $({Y_t}^{(N)})_{t\in \mathbb{R}^+}$, which can be read in \eqref{gen}, we compute the following expectation (the last line is given by Proposition \ref{momentlibre})
\begin{align*}
\mathbb{E}\left[\frac{1}{N}\Tr\left({Y_t}^{(N)}\right)\right]&=\frac{1}{N}\Tr\left(\mathbb{E}\left[{Y_t}^{(N)}\right]\right)\\
&=\frac{1}{N}\Tr\exp\left(\left[t\ L_N(Id_{U(N)})\right](I_N)\right)\\
&=\frac{1}{N}\Tr \exp\left(i\alpha t I_N-b t/2 I_N+t  \int_{U(N)} \left(g-I_N-i\Im(g)\right)\diff v_N(g)\right)\\
&=e^{t\cdot \left(i\alpha -b/2+\int_{\mathbb{U}}(\Re(\zeta)-1)\diff v(\zeta)\right)}\\
&=\tau(y_t).
\end{align*}
 Thanks to Theorem \ref{Existence Approximation Levy} of  approximation of free Lévy processes we can now prove  Theorem \ref{Existence Approximation}.
 \begin{proof}[Proof of  Theorem \ref{Existence Approximation}] Consider a free planar Markovian holonomy field in $(\mathcal{A},\tau)$ and $(y_t)_{t\ge 0}$ a free Lévy process in $(\mathcal{A},\tau)$ associated to it. According to Theorem \ref{Existence Approximation Levy}, there exists a $U(N)$-valued Lévy process $(Y^{(N)}_t)_{t\ge 0}$, converging in non-commutative distribution towards $(y_t)_{t\ge0}$ and such that  
$$\mathbb{E}[\tr_N(Y_t)]=\tau(y_t)$$
for all $t\ge 0.$ If $(\alpha,b,v)$ denotes the characteristic triplet of $(y_t)_{t\ge 0},$ Proposition \ref{momentlibre} says that
$$\tau(y_t)=e^{t\cdot \left(i\alpha -b/2+\int_{\mathbb{U}}(\Re(\zeta)-1)\mathrm{d}v(\zeta)\right)}$$
and there exists $C>0$, such that  for all $t\ge 0,$   $1-\Re(\tau(y_t))\le C t.$ 
Let us now choose $(H_l)_{l\in \sf L_0}$ to be the planar Markovian holonomy field associated to $(Y_t^{(N)})_{t\ge 0}.$  For any simple loop $s\in \sf L_{Aff}$ bordering a domain of area $t$, $\tr_N(H_t)$ has same law as $Y_t$ so that  
$$1-\mathbb{E}[\Re\left(\tr_N(H_s)\right)]=1-\Re(\tau(y_t))\le C t.$$ We can now apply Theorem \ref{extension of convergences} to conclude.
\end{proof}

\bibliographystyle{plain}
\bibliography{biblio}

\begin{thebibliography}{10}

\bibitem{Anshelevich2012a}
Michael Anshelevich and Ambar~N. Sengupta.
\newblock {Quantum free Yang-Mills on the plane}.
\newblock {\em Journal of Geometry and Physics}, 62(2):330--343, feb 2012.

\bibitem{Applebaum1993}
David Applebaum and Hiroshi Kunita.
\newblock {L{\'{e}}vy flows on manifolds and L{\'{e}}vy processes on Lie
  groups}.
\newblock {\em Kyoto Journal of Mathematics}, 33(4):1103--1123, 1993.

\bibitem{Bercovici1992}
Hari Bercovici and Dan Voiculescu.
\newblock {L{\'{e}}vy-Hin{\v{c}}in type theorems for multiplicative and
  additive free convolution.}
\newblock {\em Pacific Journal of Mathematics}, 153(2):217--248, 1992.

\bibitem{Biane1997a}
Philippe Biane.
\newblock {Free Brownian motion, free stochastic calculus and random matrices}.
\newblock In {\em Free probability theory, (Waterloo ON, 1995)}, volume~12,
  pages 1--19. Amer. Math. Soc., Providence, RI, 1997.

\bibitem{Biane1997}
Philippe Biane.
\newblock {Segal-Bargmann transform, functional calculus on matrix spaces and
  the theory of semi-circular and circular systems}.
\newblock {\em Journal of Functional Analysis}, 144(1):232--286, feb 1997.

\bibitem{Biane1998}
Philippe Biane.
\newblock {Processes with free increments}.
\newblock {\em Mathematische Zeitschrift}, 227(1):143--174, jan 1998.

\bibitem{Cebron2015b}
Guillaume C{\'{e}}bron.
\newblock {Matricial model for the free multiplicative convolution}.
\newblock {\em To appear in Annals of Probability
  (http://arxiv.org/abs/1402.5286)}, feb 2014.

\bibitem{Dahlqvist2014a}
Antoine Dahlqvist.
\newblock {Free energies and fluctuations for the unitary Brownian motion}.
\newblock {\em http://arxiv.org/abs/1409.7793}, sep 2014.

\bibitem{Douglas1995}
Michael~R. Douglas.
\newblock {Large N gauge theory - Expansions and transitions}.
\newblock {\em Nuclear Physics B - Proceedings Supplements}, 41(1-3):66--91,
  apr 1995.

\bibitem{Driver1990}
B.~K. Driver.
\newblock {\em Two dimensional Euclidean quantized Yang-Mills fields}.
\newblock Colorado Springs. World Scientific, N. J., 1990.

\bibitem{Gabriel2015c}
Franck Gabriel.
\newblock {A combinatorial theory of random matrices III: random walks on
  $\mathfrak{S}(N)$, ramified coverings and the $\mathfrak{S}(N)$ Yang-Mills
  measure}.
\newblock {\em http://arxiv.org/abs/1510.01046}, oct 2015.

\bibitem{Gabriel2015b}
Franck Gabriel.
\newblock {Combinatorial theory of permutation-invariant random matrices II:
  cumulants, freeness and L{\'e}vy processes}.
\newblock {\em http://arxiv.org/abs/1507.02465}, jul 2015.

\bibitem{Holonomy}
Franck Gabriel.
\newblock {Planar Markovian holonomy fields. A first step to the
  characterization of Markovian holonomy fields.}
\newblock {\em http://arxiv.org/abs/1501.05077}, 2015.

\bibitem{Gopakumar1995}
Rajesh Gopakumar and David~J. Gross.
\newblock Mastering the master field.
\newblock {\em Nuclear Physics B}, 451(1-2):379 -- 415, 1995.

\bibitem{Gross1988}
Leonard Gross.
\newblock The {Maxwell} equations for {Y}ang-{M}ills theory.
\newblock In {\em CMS Conf. Proc}, volume~9, pages 193--203, 1988.

\bibitem{Levy2010}
Thierry L{\'e}vy.
\newblock {\em {Two-dimensional Markovian holonomy fields}}.
\newblock Number 329. Ast\'{e}risque, 2010.

\bibitem{Levy2011}
Thierry L{\'e}vy.
\newblock {The master field on the plane}.
\newblock {\em http://arxiv.org/abs/1112.2452}, dec 2011.

\bibitem{Liao2004}
Ming Liao.
\newblock {\em {L{\'{e}}vy processes in Lie groups}}.
\newblock Cambridge University Press, 2004.

\bibitem{Speicher2006}
A.~Nica and Roland Speicher.
\newblock {\em {Lectures on the Combinatorics of Free Probability}}.
\newblock Number 335. {Lecture Note Series, London Mathematical Society,
  Cambridge University Press}, 2006.

\bibitem{Sengupta1997}
Ambar~N. Sengupta.
\newblock {\em {Gauge theory on compact surfaces, Number 600}}.
\newblock American Mathematical Society, 1997.

\bibitem{Singer1995}
Isadore~M Singer.
\newblock On the master field in two dimensions.
\newblock In {\em Functional Analysis on the Eve of the 21st Century}, pages
  263--281. Springer, 1995.

\bibitem{Voiculescu1992}
D~V Voiculescu, K~J Dykema, and A~Nica.
\newblock {\em {Free random variables}}, volume~1 of {\em CRM Monograph
  Series}.
\newblock American Mathematical Society, 1992.

\bibitem{Voiculescu1991}
Dan-Virgil Voiculescu.
\newblock {Limit laws for random matrices and free products}.
\newblock {\em Invent. Math.}, 104(1):201--220, 1991.

\bibitem{Yang1954}
C.~N. Yang and R.~L. Mills.
\newblock {Conservation of isotopic spin and isotopic gauge invariance}.
\newblock {\em Physical Review}, 96(1936):191--195, 1954.

\end{thebibliography}

\end{document}